\documentclass[twoside]{article}

\usepackage[accepted]{aistats2023}
\usepackage[round]{natbib}

\usepackage[utf8]{inputenc} 
\usepackage[T1]{fontenc}    
\usepackage{hyperref}       
\usepackage{url}            
\usepackage{booktabs}       
\usepackage{amsfonts}       
\usepackage{nicefrac}       
\usepackage{microtype}      
\usepackage{xcolor}         

\usepackage{amsmath,bm}
\usepackage{amsthm}
\usepackage{mathtools}
\usepackage{centernot}
\usepackage{algorithm}
\usepackage{algorithmic}
\usepackage{comment}

\DeclareMathOperator*{\argmax}{arg\,max}

\newcommand{\E}{\operatorname{\mathbb E}}
\newcommand{\innermid}{\;\middle\lvert\;}

\newtheorem{lemma}{Lemma}

\newtheorem{corollary}{Corollary}
\newtheorem{theorem}{Theorem}

\newtheorem{definition}{Definition}
\newtheorem{assumption}{Assumption}

\newcommand\bc[1]{\left({#1}\right)}
\newcommand\cbc[1]{\left\{{#1}\right\}}
\newcommand\abs[1]{\left|{#1}\right|}
\newcommand\brk[1]{\left\lbrack{#1}\right\rbrack}
\newcommand\norm[1]{\left\|{#1}\right\|}
\newcommand{\bck}[1]{\left\langle{#1}\right\rangle}
\newcommand\scal[2]{\bck{{#1},{#2}}}
\newcommand\Erw{\mathbb{E}}
\newcommand{\Erdos}{Erd\H{o}s}
\newcommand{\Renyi}{R\'enyi}

\allowdisplaybreaks

\begin{document}

%

%

\twocolumn[

\aistatstitle{Learning Sparse Graphon Mean Field Games}

\aistatsauthor{ Christian Fabian \And Kai Cui \And  Heinz Koeppl}

\aistatsaddress{ Technische Universität Darmstadt \And  Technische Universität Darmstadt \And Technische Universität Darmstadt } ]

\begin{abstract}
  Although the field of multi-agent reinforcement learning (MARL) has made considerable progress in the last years, solving systems with a large number of agents remains a hard challenge. Graphon mean field games (GMFGs) enable the scalable analysis of MARL problems that are otherwise intractable. By the mathematical structure of graphons, this approach is limited to dense graphs which are insufficient to describe many real-world networks such as power law graphs. Our paper introduces a novel formulation of GMFGs, called LPGMFGs, which leverages the graph theoretical concept of $L^p$ graphons and provides a machine learning tool to efficiently and accurately approximate solutions for sparse network problems. This especially includes power law networks which are empirically observed in various application areas and cannot be captured by standard graphons. We derive theoretical existence and convergence guarantees and give empirical examples that demonstrate the accuracy of our learning approach for systems with many agents. Furthermore, we extend the Online Mirror Descent (OMD) learning algorithm to our setup to accelerate learning speed, empirically show its capabilities, and conduct a theoretical analysis using the novel concept of smoothed step graphons. In general, we provide a scalable, mathematically well-founded machine learning approach to a large class of otherwise intractable problems of great relevance in numerous research fields.
\end{abstract}

\section{INTRODUCTION} \label{sec:intro}
In various research areas, scientists are confronted with systems of many interacting individuals or components. Examples include neurons in the human brain (\cite{avena2018communication}, \cite{bullmore2009complex}, \cite{bullmore2012economy},  \cite{sporns2022structure}), people trading on a stock market (\cite{bakker2010social}, \cite{bian2016evolving}) or the spreading of contagious diseases among citizens of a society (\cite{newman2002spread}, \cite{pastor2015epidemic}). Due to their complexity, these systems are in general hard to model and are often controlled by using multi-agent reinforcement learning (MARL).  In the last years, the field of MARL has experienced significant progress, see \cite{canese2021multi} or \cite{yang2020overview} for an overview, but crucial open problems remain. While many approaches provide sound empirical results, they often lack a solid theoretical foundation (\cite{zhang2021multi}). Furthermore, as the number of agents in the system increases,  numerous MARL algorithms become computationally expensive and are thereby hardly scalable.

Mean field games (MFGs) (\cite{carmona2018probabilistic}, \cite{carmona2018prob}) have proven to be a valid approach for achieving both scalability and theoretical guarantees in multi-agent systems. Since they were introduced independently by \cite{huang2006large} and \cite{lasry2007mean} to address game theoretic challenges, they have become a major interest in various research fields. Extensions of the original MFG model include discrete-time formulations (\cite{cui2021learning}, \cite{saldi2018markov}), variants with major and minor agents (\cite{carmona2016probabilistic}, \cite{firoozi2020convex}, \cite{nourian2013mm}) as well as zero-sum games (\cite{choutri2019mean}, \cite{choutri2019optimal}). MFGs are based on the weak interaction principle where each individual has a negligible influence on the whole system. Besides the numerical and theoretical benefits of this principle, MFGs provide the modelling framework for various applications, such as autonomous driving (\cite{huang2020game}), cyber security (\cite{kolokoltsov2018corruption}), big data architectures (\cite{castiglione2014exploiting}), and systemic risk in financial markets (\cite{carmona2015mean}, \cite{elie2020large}). There is also some work that aims to apply MFGs to real world tasks, e.g. social networks (\cite{yang2018learning}) or swarm robotics (\cite{elamvazhuthi2019mean}, \cite{cui2022scalable}), but this field largely remains to be developed. Although our paper is of theoretical nature, its goal is to make MFGs more realistic, as we discuss below.

Both from the classical equilibrium learning perspective and the reinforcement learning (RL) perspective, MFGs are able to provide solutions for numerous challenges where other equilibrium learning or MARL algorithms become computationally intractable. Here, learning refers to both classical computation of equilibria and RL -- also known as approximate optimal control (\cite{bertsekas2019reinforcement}) -- with focus on solving complex control problems without knowing or using the model. Current RL research is addressing the approximation of Nash equilibria for stationary games (\cite{subramanian2019reinforcement}), under the occurrence of noise (\cite{carmona2019model}), using entropy regularization (\cite{anahtarci2020q}, \cite{cui2021approximately}, \cite{guo2022entropy}), leveraging Fictitious Play (\cite{cardaliaguet2017learning}, \cite{delarue2021exploration}, \cite{hadikhanloo2019finite}, \cite{mguni2018decentralised}, \cite{perrin2021generalization}, \cite{perrin2021mean}, \cite{perrin2020fictitious}), and increasing the robustness and efficiency of learning algorithms in general (\cite{guo2019learning}, \cite{guo2020general}). A learning scheme of particular interest for our paper is Online Mirror Descent (OMD) (\cite{orabona2015generalized}, \cite{srebro2011universality}). RL research has leveraged OMD  to learn MFGs (\cite{hadikhanloo2017learning}, \cite{lauriere2022scalable}, \cite{perolat2021scaling}) which ensures algorithmic scalability.

For learning applications, decision-making on graphs appears to be particularly interesting. Here, we refer to agents connected via graphical edges, as opposed to agents with states on a graph as considered e.g. in \cite{li2019efficient}. Apart from direct graphical decompositions in MARL \cite{qu2020scalable}, there has been recent research interest in MFGs on graphs. In classical MFGs each agent weakly interacts with all other agents at once which seems to be an unrealistic modelling assumption for many applications. To overcome this concern, graphon mean field games (\cite{aurell2021stochastic}, \cite{caines2019graphon}, \cite{carmona2022stochastic}, \cite{gao2017control}, \cite{gao2021linear}) (GMFGs) provide a well-established tool to model games with a graphical structure. For example, \cite{tangpi2022optimal} apply GMFGs to model investment decisions in a financial market. Also based on GMFGs, \cite{aurell2022finite} develop models on epidemics and provide the corresponding machine learning methods for outcome estimation.
So far, most of the literature has focused on MFGs on dense graphs. To the best of our knowledge, there are only a few papers that consider sparse graphs. While \cite{gkogkas2022graphop} focuses on Kuramoto-like models, \cite{lacker2022case} is concerned with linear-quadratic stochastic differential games. Finally, \cite{bayraktar2020graphon} considers systems on not-so-dense graphs but without control and without leveraging $L^p$ graphons. By utilizing $L^p$ graphons (\cite{borgs2018p}, \cite{borgs2019L}), our paper's aim is to provide a general framework for learning MFGs on sparse graphs.

Sparse power law graphs (\cite{barabasi1999emergence}, \cite{barabasi1999mean}) are of great interest for various research applications such as social networks (\cite{aparicio2015model}), software engineering (\cite{concas2007power}, \cite{louridas2008power}, \cite{wheeldon2003power}), finance (\cite{d2016complex}), or biology (\cite{nosonovsky2020scaling}). For more examples, see \cite{newman2018networks}. Although there is strong empirical evidence for power law graphs in many research fields as mentioned above, GMFGs cannot capture these naturally sparse structures. LPGMFGs and the corresponding learning methods presented in our paper provide a novel ML tool to solve such real-world problems that are otherwise intractable.
Our contributions can be summarized as follows:
(i) introducing MFGs on $L^p$ Graphons (LPGMFGs) which formalize MFGs on sparse graphs;
(ii) conducting a theoretical analysis of LPGMFGs that includes the existence of equilibria as well as convergence guarantees;
(iii) evaluating LPGMFGs on different examples empirically, especially in a multi-class agent setup;
(iv) adapting the OMD learning scheme to our setup and thereby accelerating learning speed;
(v) conducting both an empirical and a theoretical convergence analysis for the adapted OMD algorithm.
Thus, our paper provides a scalable, mathematically well founded approach for learning MARL problems on sparse graphs on the theoretical side. On the practical side, different empirical examples demonstrate the scalability of the learning method and give an impression of how models on sparse graphs are often more realistic than dense networks.

\section{$L^p$ GRAPHONS}

\paragraph{Central Concepts.}
In this section, we briefly introduce the concept of $L^p$ graphons pioneered by \cite{borgs2018p} and \cite{borgs2019L} which provide more details. $L^p$ graphons can be informally thought of as adjacency matrices for graphs with (almost) infinitely many nodes. Naturally, approximating sparse finite graphs by these $L^p$ graphons leads to the loss of some topological information, see e.g. \cite{borgs2018p}. Nevertheless, $L^p$ graphons provide an expressive asymptotic approximation of the finite case which we show both theoretically and empirically on the next pages. In contrast to standard graphons which are limited to dense graphs, $L^p$ graphons are more general and also apply to sparse graphs. An $L^p$ graphon is a symmetric, integrable function $W\colon \brk{0,1}^2 \rightarrow \mathbb{R}$ with $\norm{W}_p < \infty$  where the $L^p$ norm on graphons is $\norm{W}_p \coloneqq  \bc{\int_{\brk{0,1}^2} \abs{W(x,y)}^p \, \mathrm d x \, \mathrm d y }^{1/p}$ for $1 \leq p < \infty$ and the essential supremum if $p = \infty$.

To quantify whether a graphon is a good approximation for a sequence of finite graphs, we associate every finite graph $G = \bc{V(G), E(G)}$ with a graphon $W^G$. For a graph $G$ with $N$ nodes, we partition the unit interval $[0,1]$ into $N$ intervals $I_1, \ldots, I_N$ of equal length. Then, the function $W^G$ is assigned a constant value on each square $I_i \times I_j$ ($i, j \in V(G)$) which is equal to one if there is an edge between the nodes $i,j$ in $G$ and zero otherwise. Thus, $W^G$ is by construction a step-function and therefore often called step-graphon. To compare some graph $G$ and graphon $W$, we can simply compare the graphons $W$ and $W^G$ in the space of graphons. Instead of $W^G$ itself, we frequently consider the normalized associated graphon $W_G / \norm{G}_1$ to derive results that also hold in the sparse case. To measure how close two graphons are, the cut norm is a natural candidate and possesses many useful properties, see e.g. \cite{lovasz2012large} for a detailed discussion. For a graphon $W\colon \brk{0,1}^2 \rightarrow \mathbb{R}$, define the cut norm by 
\[
\norm{W}_\square \coloneqq \sup_{S,T \subseteq \brk{0,1}} \abs{\int_{S \times T} W(x,y) \, \mathrm dx \, \mathrm dy},
\]
 where $S$ and $T$ range over the measurable subsets of $\brk{0,1}$.

\begin{figure*}
    \centering
    \includegraphics[width=0.99\linewidth]{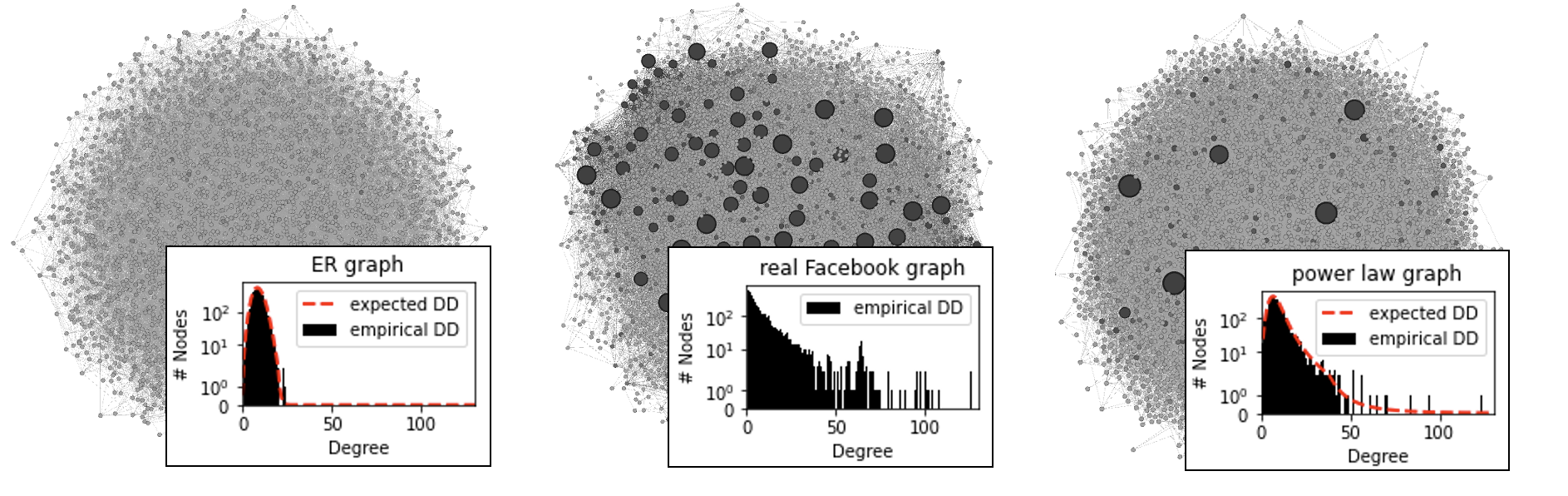}
    \caption{Three networks and their empirical and mathematically expected degree distributions (DD): \Erdos-\Renyi\ graph (left), real-world Facebook network (middle, data from \cite{rozemberczki2019gemsec}), power law graph (right, for expected DD see \cite{bollobas2007phase}); highly connected nodes are larger and darker: the Facebook graph has some nodes with high degrees and small and medium degrees otherwise. The power law graph generated by an $L^p$ graphon is qualitatively similar. All nodes in the ER graph generated by a standard graphon have degrees smaller than thirty which contradicts the real-world network. Other standard graphons, e.g. ranked attachment (\cite{borgs2011limits}), yield similar results as the ER graph but are omitted for space reasons. Each graph consists of 3892 nodes and around 17500 edges to match the real data set.}
    \label{fig:network-compare}
\end{figure*}

Starting with a graphon, we can use the following well-established construction to generate sparse random graphs with $N$ nodes where we assume without loss of generality that the vertices are labeled by $1, \ldots, N$. We choose $x_1, \ldots, x_N$ i.i.d. uniformly at random in $[0,1]$ and fix some $\rho > 0$. For all vertex pairs $1 \leq i < j \leq N$ there is an edge between $i$ and $j$ with probability $\min \cbc{\rho W(x_i, x_j), 1}$ which yields a sparse random graph $\mathbf{G} \bc{N, W, \rho}$. A sequence of sparse random graphs generated by this method converges to the generating graphon $W$ in the cut norm, see \cite[Theorem 2.14]{borgs2019L}.

\begin{assumption} \label{ass:W}
The sequence of normalized step-graphons $(W_N)_{N \in \mathbb N}$ converges in cut norm $\left\Vert \cdot \right\Vert_{\square}$ or equivalently in operator norm $\left\Vert \cdot \right\Vert_{L_\infty \to L_1}$ (see \cite{lovasz2012large}) as $N \to \infty$ to some graphon $W \in \mathcal W_0$, i.e.
\begin{align} \label{eq:Wconv}
    \left\Vert \rho_N^{-1} W_N  - W \right\Vert_{\square} \to 0, \quad \left\Vert \rho_N^{-1} W_N - W \right\Vert_{L_\infty \to L_1} \to 0 \, .
\end{align}
\end{assumption}

The limiting graphon in Assumption \ref{ass:W} is only guaranteed to exist for so-called $L^p$ upper regular graph sequences, for details see \cite{borgs2018p}, \cite{borgs2019L}. This implicit assumption means that our approach does not apply to arbitrary sequences of sparse random graphs. If the average degree in the graph sequence does not tend to infinity as $N \to \infty$, $L^p$ upper regularity is not fulfilled. Thus, for example, the asymptotic behavior of ultra-sparse graph sequences is beyond the scope of $L^p$ graphons.
Nevertheless, $L^p$ graphons are the limit of crucial sparse graph sequences such as power law graphs which cannot be provided by standard graphons. In our paper, 'power law' refers to the tail of the distribution.
Figure~\ref{fig:network-compare} shows the advantages of $L^p$ graphons over standard graphons using an exemplary real-world network (data from \cite{rozemberczki2019gemsec}). 
The examples in the next sections are usually based on power law graphons $W \colon \brk{0,1}^2 \rightarrow \mathbb{R}$ with $W(x,y) = (1 - \alpha)^2 (x y)^{- \alpha}$ where $\alpha \in (0,1)$, see \cite{borgs2018p} for details.

\paragraph{Smoothing Step Graphons.} 
For the theoretical analysis of the OMD algorithm in the next sections, we introduce the concept of smoothed step graphons which is new to the best of our knowledge. The basic idea is to smooth the borders of the steps. Then, the smoothed step graphon is Lipschitz continuous but still close to the original step graphon as we decrease the width $\xi$ of the border regions.

Consider an arbitrary step graphon $W_s$ on the unit interval partitioned into $M$ subintervals $\mathcal{I}_1, \ldots, \mathcal{I}_M$ of equal length $1/M$ such that $W_s (x, y) = w_{i, j} \geq 0$ for all $(x,y) \in \mathcal{I}_i \times \mathcal{I}_j, 1 \leq i, j \leq M$. Then, for an arbitrary but fixed $0 < \xi < 1/(2M)$, we define the corresponding smoothed step graphon $W_{s, \xi}$ as follows. For all $(x,y) \in \left\{(x,y) \in [0,1]^2: (x \leq \xi) \lor (x \geq 1 -\xi) \right\}$ and $(x,y) \in \left\{ (y \leq \xi) \lor (y \geq 1 -\xi) \right\}$ we define $W_{s, \xi} (x,y) \coloneqq W_s (x,y)$. The values of the two graphons are also defined to be identical for $(x,y)$ with $(x,y) \in \widetilde{\mathcal{I}}_i \times \widetilde{\mathcal{I}}_j$ for some $1 \leq i, j \leq M$ where $\widetilde{\mathcal{I}}_i \coloneqq [(i-1)/M + \xi, i/M - \xi)$ for all $1 \leq i \leq M$. In contrast to that, if $x \in \widetilde{\mathcal{I}}_i$ and $y \in [ j/M - \xi, j/M + \xi)$ for some $1 \leq i,j \leq M$, we have $W_{s, \xi} (x,y)   \coloneqq (\frac{1}{2} - \frac{y - j/M}{2 \xi}) w_{i, j} + \frac{y - j/M + \xi}{2 \xi} w_{i, j+1}$ and analogously for $x$ and $y$ with switched roles. Finally, if both $x \in [ i/M - \xi, i/M + \xi)$ and $y \in [ j/M - \xi, j/M + \xi)$ for some $1 \leq i,j \leq M$, we define
\begin{align*}
    W_{s, \xi} (x,y)   \coloneqq  &\left(\frac{1}{2} - \frac{x - i/M}{2 \xi}\right) \left(\frac{1}{2} - \frac{y - j/M}{2 \xi} \right) w_{i, j} \\
    & + \left(\frac{1}{2} - \frac{x - i/M}{2 \xi} \right) \frac{y - j/M + \xi}{2 \xi} w_{i, j+1} \\
    &+  \frac{x - i/M + \xi}{2 \xi} \left(\frac{1}{2} - \frac{y - j/M}{2 \xi} \right) w_{i + 1, j} \\
    &+ \frac{x - i/M + \xi}{2 \xi} \cdot \frac{y - j/M + \xi}{2 \xi}  w_{i + 1, j + j} \, .
\end{align*}
Note that $W_{s, \xi}$ is Lipschitz continuous and that by construction we have $\norm{W_{s, \xi} - W_s}_\square \leq 4 M \xi \cdot \max_{i, j} w_{i,j} $ which approaches zero as $\xi \to 0$.

\section{MODEL} \label{sec:model}

\paragraph{Finite Agent Model.} 
For the finite case, we assume that there are $N$ agents with finite state and action spaces $\mathcal{X}$ and $\mathcal{U}$, respectively. The agents implement actions at time points $\mathcal{T} \coloneqq \cbc{0, \ldots, T - 1}$ with terminal time point $T$. The interactions between individuals are modeled by a graph $G_N = (V_N, E_N)$ where each vertex represents one agent and each edge a connection between two agents. For an arbitrary finite set $A$ we denote by $\mathcal{P} (A)$ the set of all probability measures on $A$ and by $\mathcal{B} (A)$ the set of all bounded measures on $A$. Thus, the space of policies is defined as $\Pi \coloneqq \mathcal{P}(\mathcal{U})^{\mathcal{T} \times \mathcal{X}}$ and a policy of agent $i$ is denoted by $\pi^i = \bc{\pi^i_t}_{t \in \mathcal{T}} \in \Pi$ correspondingly. Furthermore, agents in the model are assumed to implement Markovian feedback policies such that they only consider local state information. Formally, this is captured by defining for all $t \in \mathcal T$ and  $i \in V_N$ the model dynamics
\begin{align}
	U^i_t \sim \pi^i_t(\cdot \mid X^i_t), \quad X^i_{t+1} \sim P(\cdot \mid X^i_t, U^i_t, \mathbb G^i_t)
\end{align}
with $X^i_0 \sim \mu_0$ for some transition kernel $P \colon \mathcal{X} \times \mathcal{U} \times \mathcal{B}(\mathcal{X}) \to \mathcal{P}(\mathcal{X})$ and the neighborhood state distribution
\begin{align}
	\mathbb G^i_t \coloneqq \frac 1 {N\rho_N} \sum_{j \in V_N} \boldsymbol{1}_{\cbc{i j \in E_N}} \delta_{X^j_t}
\end{align}
for each agent $i$ with $\delta$ being the Dirac measure and $\mathbb G^i_t \in \mathcal{B}(\mathcal{X})$ for all $i \leq N$ by definition. The normalization factor $\rho_N$ ensures that the neighborhood distribution does not converge to a vector of zeros as $N$ approaches infinity. Here, $\rho_N$ is assumed to have the same asymptotic order as the edge density of $G_N$, i.e. $\rho_N = \Theta \bc{\abs{E_N} / N^2}$ as $N \to \infty$. Each agent faces a reward function $r \colon  \mathcal{X} \times \mathcal{U} \times \mathcal{B}(\mathcal{X}) \to \mathbb{R}$ which yields her reward depending on her state, action, and the state distribution of her neighbors. Agents competitively try to maximize their expected sum of rewards
\begin{align}
	J_i^N(\pi^1, \ldots, \pi^N) \coloneqq \E \left[ \sum_{t=0}^{T-1} r \bc{X_{t}^i, U_{t}^i, \mathbb G^i_t} \right].
\end{align}
Note that we can handle the infinite-horizon discounted objective case analogously, see e.g. \cite{cui2021learning}. Finding equilibria for this type of model requires a suitable equilibrium concept. Although the classical Nash equilibrium notion (see, e.g. \cite{carmona2018probabilistic}) seems to be a natural candidate, its definition requires that no agent has an incentive to unilaterally deviate from the current policy. As we are primarily interested in an approximation via $L^p$ graphons, this equilibrium concept is to strict. Even in the limit $N \to \infty$ there can always occur (small) subgroups of agents whose graph connections deviate from the structure of the underlying graphon. Therefore, we work with the $(\epsilon, p)$-Markov-Nash equilibrium (see, for example \cite{carmona2004nash}, \cite{elie2020convergence}, \cite{cui2021learning}) which only requires optimality for a fraction $1 - p$ of all individuals. This fraction will increase, $(1 - p) \to 1$ as $N \to \infty$.

\begin{definition}
An $(\epsilon, p)$-Markov-Nash equilibrium (MNE) for $\epsilon, p > 0$ is defined as a tuple of policies $\boldsymbol{\pi} = \bc{\pi^1, \ldots, \pi^N} \in \Pi^N$ such that for any $i \in \mathcal W_N$ we have
\begin{align}
	J_i^N \bc{\boldsymbol{\pi}} \geq \sup_{\pi \in \Pi} J_i^N \bc{\pi^1, \ldots,  \pi^{i -1}, \pi,  \pi^{i+1}, \ldots,  \pi^N} - \epsilon
\end{align}
for some $\mathcal W_N \subseteq V_N$ with $\abs{\mathcal W_N} \geq \lfloor (1-p) N \rfloor$ such that $\mathcal W_N$ contains at least $\lfloor (1-p) N \rfloor$ agents.
\end{definition}

\paragraph{Mean Field Model.} 
The $L^p$ graphon mean field model (LPGMFG) constitutes the limit of the finite agent model as $N \to \infty$ and provides a reasonable approximation for the finite case. Before we formalize this claim and provide rigorous statements, we introduce the LPGMFG itself. The main difference to the $N$-agent model is that we now consider an infinite number of agents $\alpha \in \mathcal{I} \coloneqq[0,1]$. Thus, $\mathcal{M}_t \coloneqq \mathcal{P} (\mathcal{X})^{\mathcal{I}}$ denotes the space of measurable state marginal ensembles at time $t$, and $\boldsymbol{\mathcal{M}}\coloneqq \mathcal{P} (\mathcal{X})^{\mathcal{I} \times \mathcal{T}}$ the space of measurable mean field ensembles. Here, measurable means that $\alpha \mapsto \mu_t^\alpha (x)$ is measurable for all $\boldsymbol{\mu} \in \boldsymbol{\mathcal{M}}, t \in \mathcal{T}, x \in \mathcal{X}$. Analogously, a space of uniformly Lipschitz, measurable policy ensembles $\boldsymbol{\Pi} \subseteq \Pi^\mathcal{I}$ is defined such that $\alpha \mapsto \pi^\alpha_t (u \vert x)$ is measurable and $L_{\boldsymbol \Pi}$-Lipschitz for any $\boldsymbol{\pi} \in \boldsymbol{\Pi}, t \in \mathcal{T}$, $u \in \mathcal{U}, x \in \mathcal{X}$. Intuitively, a policy ensemble $\boldsymbol{\pi} \in \boldsymbol{\Pi}$ includes an infinite number of policies $\pi^\alpha \in \Pi$ where each policy is associated with one agent $\alpha$. State and action variables are defined for all $(\alpha, t) \in \mathcal I \times \mathcal T$ as
\begin{align} \label{eq:dynamics}
	 U^\alpha_t \sim \pi^\alpha_t(\cdot \mid X^\alpha_t), \quad X^\alpha_{t+1} \sim P(\cdot \mid X^\alpha_t, U^\alpha_t, \mathbb G^\alpha_t), 
\end{align}
with $X^\alpha_0 \sim \mu_0$ where the deterministic neighborhood MF of agent $\alpha$ for some deterministic MF $\boldsymbol{\mu} \in \boldsymbol{\mathcal{M}}$ is
\begin{align} \label{eq:empneighbormf}
	\mathbb G^\alpha_t \coloneqq \int_{\mathcal I} W(\alpha, \beta) \mu^\beta_t \, \mathrm d\beta
\end{align}
 with $\mathbb G^\alpha_t \in \mathcal{B} (\mathcal{X})$ by definition. Each agent tries to competitively maximize her rewards given by
\begin{align}
	J^{\boldsymbol \mu}_{\alpha}\bc{\pi^\alpha} \coloneqq \E \left[ \sum_{t=0}^{T-1} r \bc{X_{t}^\alpha, U_{t}^\alpha, \mathbb G^\alpha_t} \right].
\end{align}
Now, it remains to adapt the Nash equilibrium concept to the LPGMFG case.  Thus, we introduce two functions $\Psi \colon \boldsymbol{\Pi} \to \boldsymbol{\mathcal{M}}$ and $\Phi \colon \boldsymbol{\mathcal{M}} \to 2^{\boldsymbol{\Pi}}$. $\Psi$ maps a policy ensemble $\boldsymbol{\pi} \in \boldsymbol{\Pi}$ to the mean field ensemble $\boldsymbol{\mu} = \Psi (\boldsymbol{\pi}) \in \boldsymbol{\mathcal{M}}$ generated by $\boldsymbol{\pi}$ which is formalized by the recursive equation
\begin{align}
	\mu_{t+1}^\alpha \bc{x} = \sum_{\substack{x' \in \mathcal X \\ u \in \mathcal U}} \mu_{t}^\alpha \bc{x'} \pi_t^\alpha \bc{u \vert x'} P \bc{x \vert x', u, \mathbb G^\alpha_t}
\end{align}
for all $\alpha \in [0,1]$ with $\mu_0^\alpha \equiv \mu_0$. The second map $\Phi\colon \boldsymbol{\mathcal{M}} \to 2^{\boldsymbol{\Pi}}$ takes a mean field ensemble $\boldsymbol{\mu} \in \boldsymbol{\mathcal{M}}$ and maps it to the set of  policy ensembles $\Phi \bc{\boldsymbol{\mu}} \subseteq 2^{\boldsymbol{\Pi}}$ that are optimal with respect to $\boldsymbol{\mu}$, i.e. $\pi^\alpha = \argmax_{\pi \in \Pi} J^{\boldsymbol \mu}_{\alpha}\bc{\pi^\alpha}$ for all $\alpha \in [0,1]$.
With the above definitions, we can state the equilibrium concept for LPGMFGs, namely the $L^p$ graphon mean field equilibrium (GMFE).
\begin{definition}
	A GMFE is a tuple $\bc{\boldsymbol{\mu}, \boldsymbol{\pi}} \in \boldsymbol{\Pi} \times \boldsymbol{\mathcal M}$ such that $\boldsymbol{\pi} \in \Phi 
	\left(\boldsymbol{\mu}\right)$ and $\boldsymbol{\mu} = \Psi (\boldsymbol{\pi})$.
\end{definition}
We also refer to the policy part of a GMFE as its Nash Equilibrium (NE). Intuitively, a GMFE consists of a policy ensemble $\boldsymbol{\pi}$ and a MF ensemble $\boldsymbol{\mu}$ such that $\boldsymbol{\pi}$ generates $\boldsymbol{\mu}$ and is also an optimal response to the generated MF.
We will frequently use a Lipschitz assumption common in the MFG literature (\cite{bayraktar2020graphon},  \cite{carmona2018probabilistic}, \cite{cui2021learning}) to enable the derivation of expressive theoretical results. The power law graphon, however, is not Lipschitz, so we derive a Lipschitz cutoff version in Appendix~\ref{app:cut_off}. Since this cutoff power law graphon does not yield qualitatively different results compared to the power law graphon, it is omitted from the main text.
\begin{assumption} \label{ass:Lip}
	Let $r$, $P$, $W$ be Lipschitz continuous with Lipschitz constants $L_r, L_P, L_W > 0$, or alternatively there exist $L_W > 0$ and disjoint intervals $\{\mathcal I_1, \ldots,\mathcal I_Q\}$, $\cup_{i} \mathcal I_i = \mathcal I$ s.t. $\forall i,j \leq Q$ and $\forall (x,y),(\tilde x,\tilde y) \in \mathcal I_i \times \mathcal I_j$,
	\begin{align} \label{eq:blockwiseLip}
		|W(x,y) - W(\tilde x, \tilde y)| \leq L_W (|x - \tilde x| + |y - \tilde y|).
	\end{align}
\end{assumption}

Under Assumption~\ref{ass:Lip}, the model defined above has a GMFE which is formalized by the next theorem.

\begin{theorem} \label{thm:existence}
Under Assumption~\ref{ass:Lip} and for Lipschitz $W$, there exists a GMFE $(\boldsymbol \pi, \boldsymbol \mu) \in \boldsymbol \Pi \times \boldsymbol{\mathcal M}$.
\end{theorem}

\begin{proof}
    The existence of a GMFE follows from \cite[Theorem~3.3]{saldi2018markov} for the extended state space $\mathcal{X} \times [0,1]$. See also \cite[Proof of Theorem~1]{cui2021learning}.
\end{proof}

\paragraph{Mean Field Approximation.}
The proofs of all theoretical results can be found in the Appendix. This paragraph relates the finite agent model to the MF model by showing that LPGMFGs yield an increasingly accurate approximation for the $N$-agent case as the number of agents grows. We emphasize that the LPGMFG yields an approximation for the $N$-agent game for all $N$ at once. Both the theoretical results as well as the empirical findings show that the accuracy of this approximation increases with the number $N$ of agents, see the next sections for details. As a consequence, the LPGMFG concept provides a scalable and increasingly accurate approximation for otherwise intractable multi agent problems with a large number of individuals. By Theorem~\ref{thm:existence}, there exists a GMFE $(\boldsymbol \pi, \boldsymbol \mu)$ which yields an approximate NE for the $N$-agent problem through the map $\Gamma_N (\boldsymbol{\pi}) \coloneqq \bc{\pi^1, \ldots, \pi^N} \in \Pi^N$  defined by $\pi_t^i (u \vert x) \coloneqq \pi_t^{\alpha_i} (u \vert x)$ for all $\alpha \in \mathcal{I}$, $t \in \mathcal{T}$, $x \in \mathcal{X}, u \in \mathcal{U}$ with $\alpha_i = i/N$.

For a theoretical comparison, we lift both the policies and empirical distributions in the finite agent model to the continuous domain $\mathcal{I}$. Thus, for an $N$-agent policy tuple $\bc{\pi^1, \ldots, \pi^N} \in \Pi^N$  the corresponding step policy ensemble $\boldsymbol{\pi}^N \in \boldsymbol{\Pi}$  and the random empirical step measure ensemble $\boldsymbol{\mu}^N \in \boldsymbol{\mathcal{M}}$ are defined by $\pi_t^{N, \alpha} \coloneqq \sum_{i \in V_N} \boldsymbol{1}_{\cbc{\alpha \in (\frac{i-1}{N}, \frac{i}{N}]}} \cdot \pi_t^i$ and $\mu_t^{N, \alpha} \coloneqq \sum_{i \in V_N} \boldsymbol{1}_{\cbc{\alpha \in (\frac{i-1}{N}, \frac{i}{N}]}} \cdot \delta_{X_t^j}$
for all $\alpha \in \mathcal{I}$ and $t \in \mathcal{T}$. For notational convenience, we furthermore define for any $f\colon \mathcal{X} \times \mathcal{I} \to \mathbb{R}$ and state marginal ensemble $\boldsymbol{\mu}_t \in \mathcal{M}_t$, $\boldsymbol{\mu}_t \bc{f} \coloneqq \int_{\mathcal{I}} \sum_{x \in \mathcal{X}} f (x, a) \mu_t^\alpha (x) \, \mathrm d \alpha$.
With these definitions in place, we state our first main theoretical result.

\begin{theorem} \label{thm:muconv}
Consider $\boldsymbol \pi \in \boldsymbol \Pi$ with $\boldsymbol \mu = \Psi(\boldsymbol \pi)$. Under Assumption~\ref{ass:W} and the $N$-agent policy $(\pi^1, \ldots, \pi^{i-1}, \hat \pi, \pi^{i+1}, \ldots, \pi^N) \in \Pi^N$ with $(\pi^1, \pi^2, \ldots, \pi^N) = \Gamma_N(\boldsymbol \pi) \in \Pi^N$, $\hat \pi \in \Pi$, $t \in \mathcal T$, we have for all measurable functions $f \colon \mathcal X \times \mathcal I \to \mathbb R$ uniformly bounded by some $M_f > 0$, that
\begin{align}
    &\E \left[ \left| \boldsymbol \mu^N_t(f) - \boldsymbol \mu_t(f) \right| \right] \to 0
\end{align}
uniformly over all possible deviations $\hat \pi \in \Pi, i \in V_N$. If the graphon convergence in Assumption~\ref{ass:W} is up to rate $ O(1/\sqrt N)$, then this rate of convergence is the same.
\end{theorem}
Based on Theorem~\ref{thm:muconv}, we can derive a central result of this paper which formalizes the capability of LPGMFGs to approximate finite $N$-agent models. In contrast to prior work, proving Theorem~\ref{thm:muconv} requires an additional mathematical effort which is discussed in Appendix~\ref{sec:appendix_overview}.
\begin{theorem} \label{thm:approxnash}
Consider a GMFE $(\boldsymbol \pi, \boldsymbol \mu)$ under Assumptions~\ref{ass:W} and \ref{ass:Lip}. For any $\varepsilon, p > 0$ there exists $N'$ such that for all $N > N'$, the policy $ \Gamma_N(\boldsymbol \pi) \in \Pi^N$ is an $(\varepsilon, p)$-MNE.
\end{theorem}
Intuitively, Theorem~\ref{thm:approxnash} states that the GMFE provides an increasingly accurate approximation of the $N$-agent problem as the number of agents goes up. 
Since the algorithmic computation of NE is in general intractable (\cite{conitzer2008new}, \cite{papadimitriou2001algorithms}, \cite{papadimitriou2007complexity}), the LPGMFGs approximation can overcome these difficulties by choosing $\epsilon$ and $p$ in Theorem~\ref{thm:approxnash} close to zero when the number $N$ of agents is sufficiently large.

\section{LEARNING LPGMFGS} \label{sec:learning}
\paragraph{Equivalence Class Method.}
For learning equilibria in LPGMFGs, we introduce equivalence classes (\cite{cui2021learning}). We discretize the continuous interval $\mathcal{I}$ of agents by some finite number $M$ of subintervals that form a partition of $\mathcal{I}$. For convenience, we usually assume that every subinterval has the same length. Then, all agents within one class, i.e. a subinterval, are approximated by the agent who is located at the center of the respective subinterval. Subsequently, we can solve the optimal control problem for each equivalence class separately by applying either backwards induction or RL. Although this formulation seems to resemble classical multi-population mean field games (MP MFGs) (\cite{huang2006large}, \cite{perolat2021scaling}) at first, the crucial advantages of LPGMFGs are that they are on the one hand rigorously connected to finite agent games. On the other hand, they can handle an uncountable number of agent equivalence classes that cannot be captured by the standard multi-class model. Beyond that, the just described learning method for LPGMFGs does not just provide an approximation for some finite $N$-agent problem with a fixed $N$. Instead, it yields an estimation for the $N$-agent problem for all arbitrary, large enough $N$ at once. The technical details of the approach can be found in Appendix~\ref{app:learning_methods}.

\paragraph{Online Mirror Descent (OMD).}
The discretized game generated by the equivalence class method can be interpreted as a MP MFG with $M$ populations. In the literature, the concept of OMD is used to learn equilibria in such MP MFGs (\cite{hadikhanloo2017learning}, \cite{perolat2021scaling}). Our paper leverages these concepts to learn LPGMFGs.

To prove convergence for the OMD algorithm, we have to ensure that a NE exists. Here, the discretized GMFG can be interpreted as a GMFG on the step-graphon $W_s$ created by discretization. To facilitate the theoretical analysis of the OMD algorithm, we consider the corresponding smoothed step graphon $W_{s, \xi}$ and the smoothed GMFG given by the dynamics $\hat{U}^\alpha_t \sim \pi^\alpha_t(\cdot \mid \hat{X}^\alpha_t)$ and $\hat{X}^\alpha_{t+1} \sim P(\cdot \mid \hat{X}^\alpha_t, \hat{U}^\alpha_t, \hat{\mathbb{G}}^\alpha_t)$,
with $\hat{X}^\alpha_0 \sim \mu_0$ for all $(\alpha, t) \in \mathcal{I} \times \mathcal{T}$ where $\hat{\mathbb{G}}$ is the neighborhood state distribution for the smoothed step graphon. One advantage of this approach is that the existence of a GMFE $(\boldsymbol{\mu}_{s, \xi}, \boldsymbol{\pi}_{s, \xi})$ is ensured by Theorem 1 for this GMFG. Also, for $\xi$ close enough to zero, the smoothed step graphon converges to the original step graphon in the cut norm. 
\begin{theorem}\label{thm:smoothed_reward}
Suppose that $(\boldsymbol{\mu}_{s, \xi}, \boldsymbol{\pi}_{s, \xi})$ is a GMFE in the smoothed version of the MP MFG on the step graphon $W$ under Assumption~\ref{ass:Lip}. Then, for every $\epsilon, p > 0$ there exists a $\xi' >0$ such that for all $0 < \xi < \xi'$ 
\begin{align}
    \sup_{\boldsymbol{\pi} \in \Pi} J_{\alpha, W}(\boldsymbol{\pi}_{s, \xi}) - J_{\alpha, W} (\boldsymbol{\pi}) \leq  \varepsilon
\end{align}
for all $\alpha \in \mathcal{J}$ for some $\mathcal{J} \subseteq \mathcal{I}$ with Lebesgue measure $\lambda (\mathcal{J}) \geq 1-p$.
\end{theorem}
This means that a GMFE in the smoothed version of the GMFG is an $(\epsilon, p)$-MNE for the discretized game. Combining this insight with existing results \cite[Theorem 5]{cui2021learning} indicates that the smoothed GMFE provides a good approximation for the finite agent case, but we leave a rigorous proof for future work.
We call a smoothed MP MFG weakly monotone if for any $\boldsymbol \pi, \boldsymbol \pi' \in \boldsymbol \Pi$ we have 
\begin{align}\label{def:mon_metric}
	\tilde{d}\bc{\boldsymbol{\pi}, \boldsymbol{\pi'}} &\coloneqq \int_{\mathcal{I}} \left[ J^{\boldsymbol \mu}_{\alpha}\bc{\pi^\alpha} + J^{\boldsymbol \mu'}_{\alpha}\bc{{\pi'}^\alpha} \right. \\
	&\qquad \qquad \left. - J^{\boldsymbol \mu}_{\alpha}\bc{{\pi'}^\alpha} -J^{\boldsymbol \mu'}_{\alpha}\bc{\pi^\alpha} \right] \mathrm d \alpha \leq 0 \nonumber
\end{align}
where $\boldsymbol \mu = \Psi(\boldsymbol \pi)$ and $\boldsymbol \mu' = \Psi(\boldsymbol \pi')$ are the MFs associated with the respective policies. If the inequality is strict $\forall \boldsymbol \pi \neq \boldsymbol \pi'$, we call the MP MFG strictly weakly monotone. 
\begin{assumption}\label{ass:swmonotone}
The smoothed MP MFG is strictly weakly monotone.
\end{assumption}
Weak monotonicity can be interpreted as agents preferring less crowded areas over crowded ones. 
Under Assumption~\ref{ass:swmonotone}, the NE guaranteed by Theorem~\ref{thm:existence} is unique.
\begin{lemma}\label{lem:omd_unique}
If the smoothed MP MFG satisfies Assumptions~\ref{ass:Lip} and \ref{ass:swmonotone}, it has a unique NE.
\end{lemma}

We define the OMD algorithm as in \cite{perolat2021scaling} and consider the continuous time case (CTOMD) where we denote the time of the algorithm by $\tau > 0$. Then, we obtain the following convergence result.
\begin{theorem}\label{thm:unique_conv}
	If the smoothed MP MFG satisfies Assumptions~\ref{ass:Lip} and \ref{ass:swmonotone} and the transition kernel does not depend on the MF, the sequence of policies $\bc{\boldsymbol{\pi}_\tau}_{\tau \geq 0}$ generated by the CTOMD algorithm converges to the unique NE as $\tau \to \infty$. 
\end{theorem}

Empirical evidence collected in our simulations suggests that a convergence guarantee as in Theorem~\ref{thm:unique_conv} also holds for the case where the MF depends on the transition kernel, but we leave a rigorous proof for future work.

\section{EXPERIMENTS} \label{sec:experiments}

\begin{figure*}
    \centering
    \includegraphics[width=0.99\linewidth]{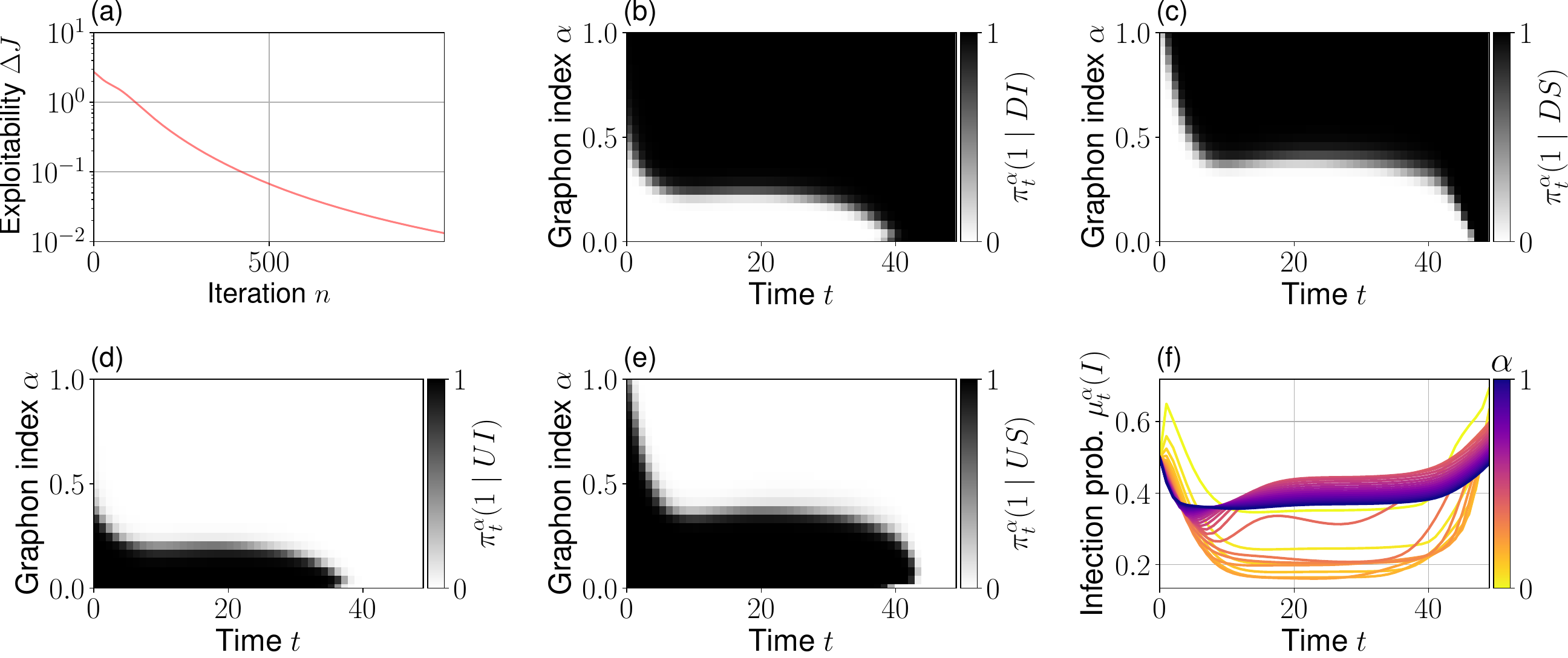}
    \caption{Experimental results for OMD on the Cyber Security problem. (a): The exploitability $\Delta J$ over iterations $n$ of OMD; (b)-(e): The probability of choosing action $u=1$ at graphon index $\alpha$ and time $t$ under the final equilibrium policy in states $DI, DS, UI, US$ respectively; (f): The probability (mean-field) of infected agents, visualized for each discretized $\alpha$.}
    \label{fig:cyber}
\end{figure*}

In our experiments, we use OMD with its hyperparameter $\gamma$ set to $1$, the power law $L^p$-graphon $W$, and discretize $\mathcal I$ into $M=25$ subintervals for the Cyber Security problem or $M=10$ for the Beach Bar problem given as follows. Here, we emphasize that using $L^p$-graphons in the experiments is a key component of our LPFGMFG approach. This allows us to model many realistic networks which are characterized by sparsity and power law degree distributions. As we discussed previously, standard GMFG approaches are conceptually unable to capture these networks.

\begin{figure*}[ht]
    \centering
    \includegraphics[width=0.99\linewidth]{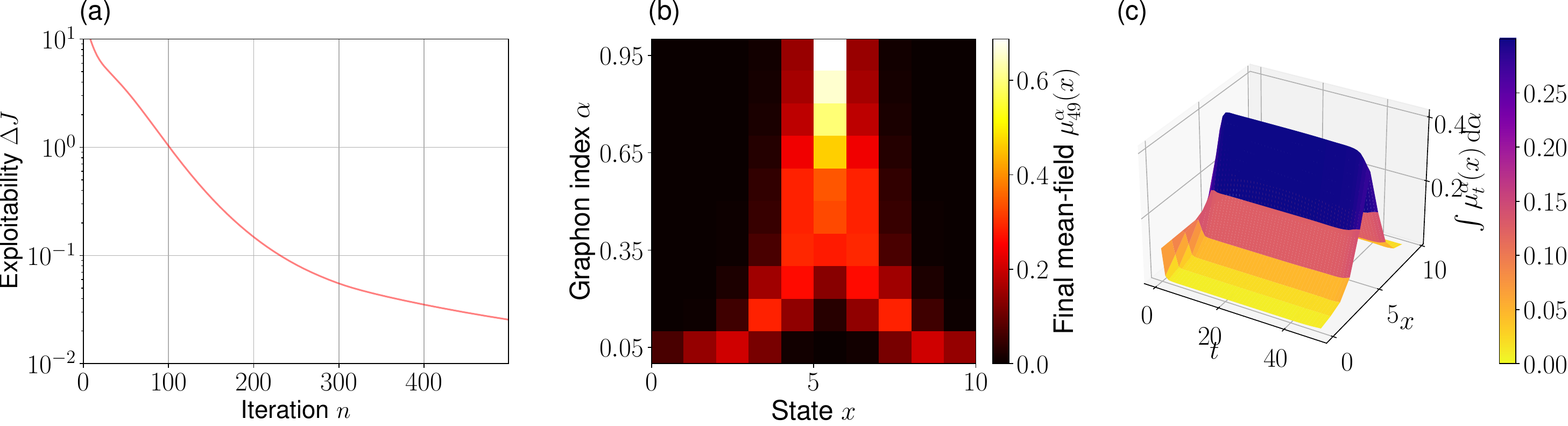}
    \caption{Experimental results for OMD on the Beach problem. (a) The exploitability $\Delta J$ over iterations $n$ of OMD; (b) The final distribution over positions on the beach at time $t=T-1$ for each discretized $\alpha$; (c) The evolution of distributions over time.}
    \label{fig:beach}
\end{figure*}

\paragraph{Cyber Security.} 
We modify an existing cyber security model (\cite{carmona2018probabilistic}, \cite{kolokoltsov2016mean}) where a virus spreads to computers either directly by an attack, or by other nearby infected computers.
In contrast to existing work, we use LPGMFGs to allow malware spread only by neighboring computers to increase the modelling accuracy. 
Each computer can be either infected ($I$) or susceptible ($S$), as well as defended ($D$) or unprotected ($U$), formally $\mathcal{X} \coloneqq \cbc{DI, DS, UI, US}$. Agents may attempt to switch (with geometrically distributed delay) between defense states, $\mathcal{U} \coloneqq \cbc{0, 1}$. The recovery and infection probabilities depend on the defense state and number of infected neighbors, while the reward function consists of costs for being defended or infected.
Details can be found in Appendix~\ref{app:cyber}.

\paragraph{Heterogeneous Cyber Security.}

A natural extension of the cyber security model is the adaptation to a multi-class framework with heterogeneous agent classes. For illustrative purposes we focus on only two types of agents -- private ($\mathrm{Pri}$) and corporate ($\mathrm{Cor}$), see Appendix~\ref{app:heterogeneous} for details.

\paragraph{Beach Bar Process.} Introduced as the Santa Fe bar problem (\cite{arthur1994inductive}, \cite{farago2002fair}), variations of the Beach Bar Process are frequently used to demonstrate the capabilities of learning algorithms (\cite{perolat2021scaling}, \cite{perrin2020fictitious}). Agents can move their towels between locations and try to be close to the bar but also avoid crowded areas and neighbors in an underlying network. Formally, we consider a one-dimensional beach bar process with $|\mathcal X| = 10$ locations $\mathcal X = \{0, 1, \ldots, |\mathcal X|-1\}$ where a bar is located in the middle $B = |\mathcal X|/2$ of the beach. The $N$ agents may move their towel between locations, $\mathcal U = \{-1, 0, 1\}$. Furthermore, the agents are connected by a power law network where connected agents try to avoid being close to each other. See Appendix~\ref{app:beach_bar} for details.

\begin{figure*}
    \centering
    \includegraphics[width=0.99\linewidth]{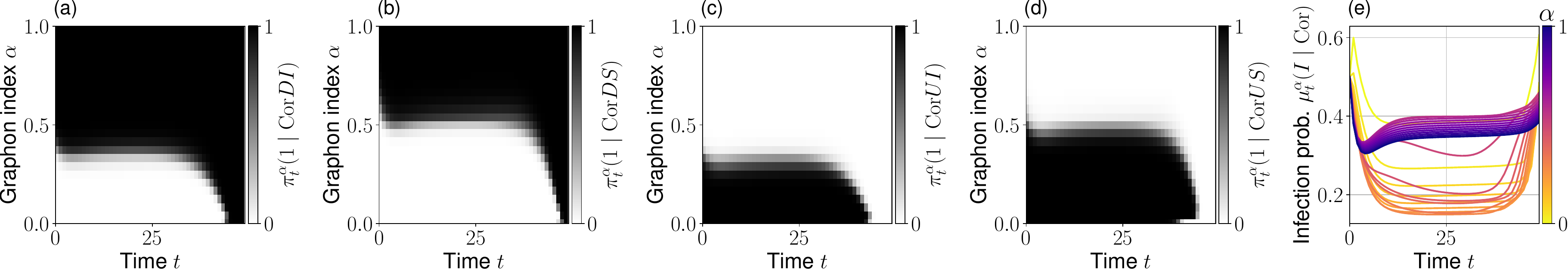}
    \includegraphics[width=0.99\linewidth]{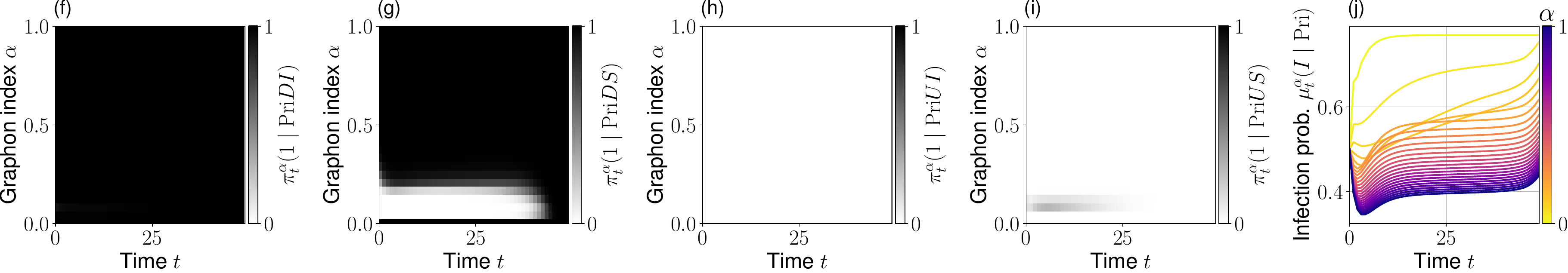}
    \caption{Experimental results for OMD on the heterogeneous Cyber Security problem. (a)-(d): The probability of action $u=1$ at graphon index $\alpha$ and time $t$ under the final equilibrium policy in states $\mathrm{Cor}DI, \mathrm{Cor}DS, \mathrm{Cor}UI, \mathrm{Cor}US$; (e): The probability (MF) of infected $\mathrm{Cor}$ agents, visualized for each discretized $\alpha$; (f)-(j): Same as in (a)-(e) but for $\mathrm{Pri}$ agents.}
    \label{fig:cyber2}
\end{figure*}

\begin{figure}
    \centering
    \includegraphics[width=0.99\linewidth]{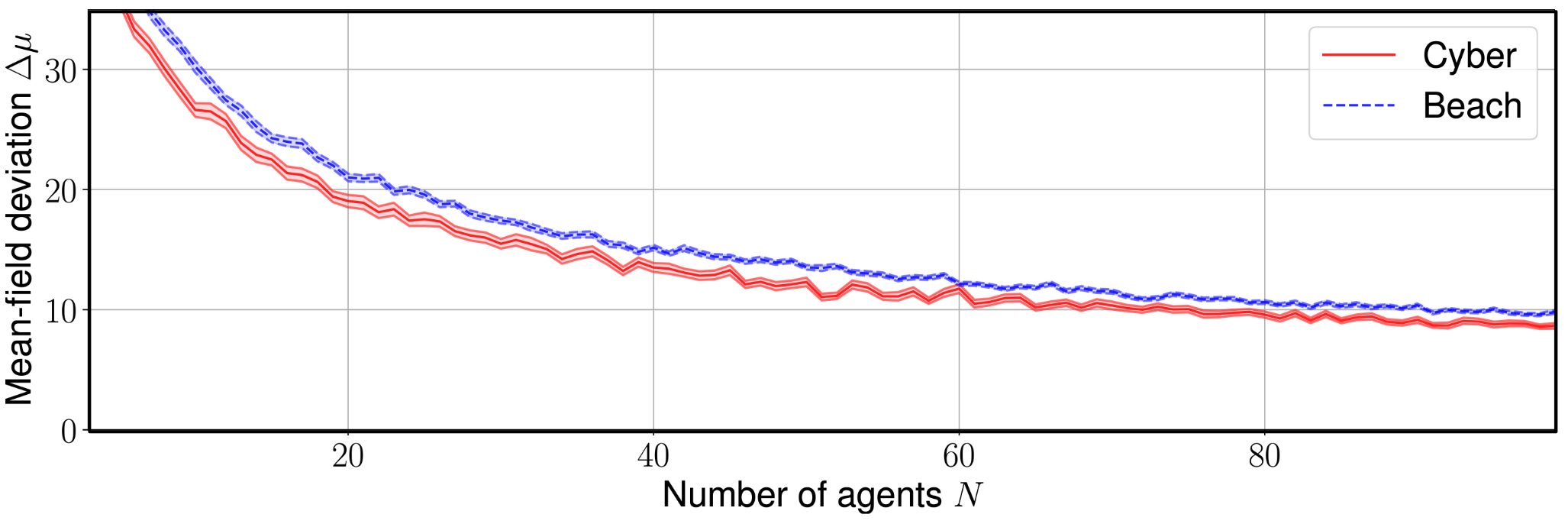}
    \caption{The $L_1$ error between the empirical distribution and the limiting MF $\Delta \mu = \mathbb E \left[ \sum_{t \in \mathcal T, x \in \mathcal X} \left| \frac 1 N \sum_{i} \delta_{X^i_t}(x) - \int_{\mathcal I} \mu^\alpha_t(x) \, \mathrm d\alpha \right| \right]$ at $\beta=0.51$ averaged over $100$ randomly sampled graphs with $N$ nodes and $68\%$ confidence interval (shaded).}
    \label{fig:n-compare}
\end{figure}

\paragraph{Experimental Results.}
As seen in Figures~\ref{fig:cyber} and \ref{fig:beach}, the approximate exploitability $\Delta J(\pi) = \int_{\mathcal I} \sup_{\pi^* \in \Pi} J^{\Psi(\pi)}_\alpha(\pi^*) - J^{\Psi(\pi)}_\alpha(\pi) \, \mathrm d\alpha$
of a MF policy $\pi$ quantifies the sub-optimality of the obtained equilibrium and quickly converges in the Cyber Security and Beach Bar scenario using OMD. We obtain near-identical results also for the Heterogeneous Cyber Security problem, which are omitted for space reasons. The algorithm converges to an equilibrium where, as expected, the agents with the most connections in the graph attempt to defend at fixed cost, as their expected cost from not defending is higher than for agents with fewer connections. The system quickly runs into an almost time-stationary state, where the costs of defending equilibrate with the expected cost of future infection, see Figure~\ref{fig:cyber}. Since we consider a finite-horizon however, the option of defending becomes increasingly unattractive as time runs out. The probability of an agent $\alpha$ being infected at any time shows an interesting behavior: At $\alpha=0$, the probability is quite high due to the great number of connections. As $\alpha$ decreases, so does the probability of being infected at all times. However, as soon as $\alpha$ passes a threshold where defense is no longer individually worth it, the fraction of infected nodes jumps up.

In the heterogeneous case, as seen in Figure~\ref{fig:cyber2}, we consider an additional class of nodes with partially similar behavior. For very high connectivity $\alpha \to 0$ however, we observe that $\mathrm{Pri}$ nodes will never defend themselves, since for the given problem parameters, the probability of infection will be very high regardless of the defense state. Otherwise, we can observe similar behavior as in the homogeneous case. Perhaps most interesting is the asymmetry between choosing to switch between defended and undefended. When agents are susceptible, some agents will opt to neither switch from defended to undefended, nor vice versa. This stems from the fact that agents switching in state $US$ could likely jump to $UI$ and $DI$, while in state $DS$ likely jumps are $DS$ and $US$, each of which may have different future returns. For the Beach Bar process in Figure~\ref{fig:beach}, we see results similar to the classical ones in \cite{perrin2020fictitious}. By giving each agent an incentive to avoid only their direct graphical neighbors, we obtain an equilibrium behavior where agents with many connections will stay further away from the bar, while agents with few connections will not mind many other agents.

\begin{figure}
    \centering
    \includegraphics[width=0.99\linewidth]{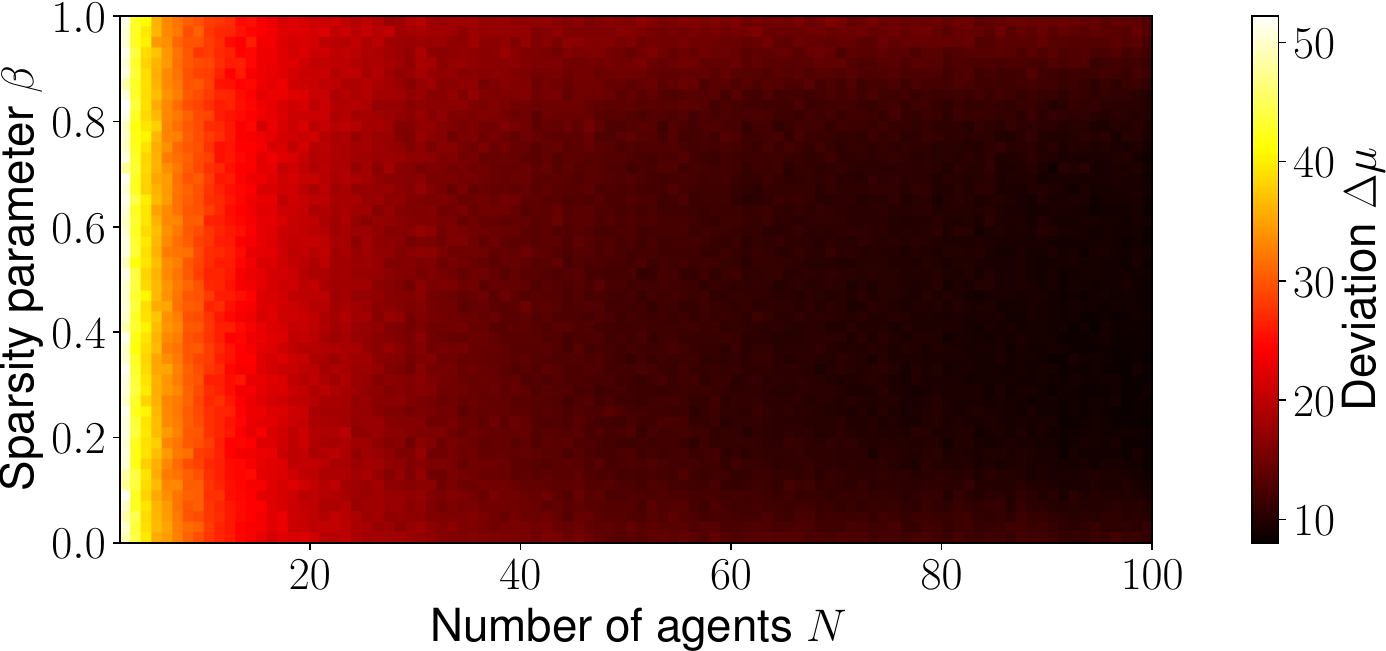}
    \caption{The $L_1$ error between the empirical distribution and the limiting MF as in Figure~\ref{fig:n-compare} over $50$ uniformly spaced $\beta \in (0,1)$ and $N \leq 100$ for the Cyber Security problem.}
    \label{fig:full-n-compare}
\end{figure}

Finally, in Figures~\ref{fig:n-compare} and \ref{fig:full-n-compare} for $\rho_N = N^{-\beta}$ and sparsity parameter $\beta \in (0,1)$, we observe convergence of the $N$-agent system objective to the MF objective, implying that sufficiently large finite systems are well-approximated by the LPGMFG. In Figure~\ref{fig:full-n-compare}, for $\beta$ close to $0$ or $1$, convergence slows down, as by \cite[Theorem 2.14]{borgs2019L} convergence is only guaranteed for $0 < \beta < 1$. Even though for $\beta = 0$, one would get the same model as in \cite{cui2021learning}, since the power law graphon is not $[0,1]$-valued, approximation guarantees fail for $\beta = 0$ and we observe increasingly slow convergence as we approach zero.

\section{CONCLUSION} \label{sec:conclusion}
In this paper we have introduced LPGMFGs which enable the scalable, mathematically sound analysis of otherwise intractable MARL problems on large sparse graphs. We rigorously derived existence and convergence guarantees for LPGMFGs and provided learning schemes to find equilibria algorithmically where we adapted the OMD learning algorithm to the setting of LPGMFGs. Beyond that, we demonstrated the benefits of our approach empirically on different examples and showed that the practical results match the theory. As for the societal impact we foresee from our work, we believe that while our techniques remain very general, they could in the future lead to an analysis of self-interested agents on real graphs such as from social networks. This could find application e.g. in control strategies for future pandemics, or other interventions. Future work could extend our model in numerous ways such as considering continuous time, action, and state spaces or adding noise terms. A challenging task could also be to find similar learning concepts for ultra-sparse graphs where the degrees remain constant as the number of agents becomes large.  For applications, it would be interesting to use LPGMFGs to solve real-world problems that occur in various research fields. In general, we hope that our work contributes to the MARL literature and inspires future work on scalable learning methods on sparse graphs.

\subsubsection*{Acknowledgements}
 This work has been co-funded by the Hessian Ministry of Science and the Arts (HMWK) within the projects "The Third Wave of Artificial Intelligence - 3AI" and hessian.AI, and the LOEWE initiative (Hesse, Germany) within the emergenCITY center.


\bibliography{main}

\newpage

\onecolumn
\appendix

\section*{SUPPLEMENTARY MATERIAL}

The following pages provide additional information and proofs for the statements in the main part of the paper "Learning Sparse Graphon Mean Field Games". In the sections~\ref{sec:appendix_overview} to \ref{sect:thm_2_3_last} we give the proofs of Theorems~\ref{thm:muconv} and \ref{thm:approxnash}. Subsequently, sections~\ref{sec:smoothed_proofs} to \ref{sec:thm_smoothed_last} contain the proofs of Theorem~\ref{thm:smoothed_reward} and the necessary intermediate results.
Furthermore, the sections \ref{app:learning_methods} to \ref{sec:final_proof_OMD} include the details of the learning methods as well as the proofs for the theoretical convergence guarantee (Theorem~\ref{thm:unique_conv}) of the OMD algorithm. We provide a derivation of the cutoff power law graphon in section~\ref{app:cut_off} and conclude the supplementary materials with the problem details for our examples in section~\ref{app:prob_details}.

\section{OVERVIEW}\label{sec:appendix_overview}

In this section we collect intermediate results and auxiliary definitions that help to prove Theorems~\ref{thm:muconv} and \ref{thm:approxnash}. Proofs which require an extensive argumentation are deferred to the subsequent paragraphs.
First, we define the $\alpha$-neighborhood map $\mathbb G^\alpha \colon \mathcal M_t \to \mathcal P(\mathcal X)$ of an agent $\alpha \in \mathcal{I}$ by 
\begin{align*}
    \mathbb G^\alpha(\boldsymbol \mu_t) \coloneqq \int_{\mathcal I} W(\alpha, \beta) \mu^\beta_t \, \mathrm d\beta
\end{align*}
and similarly the empirical $\alpha$-neighborhood map $\mathbb G^\alpha_N \colon \mathcal M_t \to \mathcal P(\mathcal X)$ by
\begin{align*}
 \mathbb G^\alpha_N(\boldsymbol \mu_t) \coloneqq \int_{\mathcal I} \frac{W_N(\alpha, \beta)}{\rho_N} \mu^\beta_t \, \mathrm d\beta.
\end{align*}
Here, we point out that $\mathbb G^\alpha_t = \mathbb G^\alpha(\boldsymbol \mu_t)$  holds for the MF system and $\mathbb G^i_t = \mathbb G^{\frac i N}_N(\boldsymbol \mu^N_t)$ holds for the finite model.
For the following proofs, we require the concept of an ensemble transition kernel operator $P_{t, \boldsymbol{\nu}, W}^{\boldsymbol{\pi}}$ which is defined by
\begin{align*}
    \bc{\boldsymbol{\nu} P_{t, \boldsymbol{\nu}', W}^{\boldsymbol{\pi}}}^\alpha &= \sum_{x \in \mathcal X} \nu^\alpha (x) \sum_{u \in \mathcal{U}} \pi_t^\alpha (u \vert x) \cdot P \bc{\cdot \, \bigg\vert x, u, \int_{\mathcal{I}} W \bc{\alpha, \beta} {\nu'}^\beta \, \mathrm d \beta}
\end{align*}
where $W$ is a graphon, $\boldsymbol{\nu}, \boldsymbol{\nu}' \in \mathcal{M}_t$, and $\boldsymbol{\pi} \in \boldsymbol{\Pi}$. This definition especially implies the useful property $\boldsymbol{\mu}_{t+1} = \boldsymbol{\mu}_{t} P_{t, \boldsymbol{\mu}_t, W}^{\boldsymbol{\pi}}$. For convenience, we also introduce the notation
\begin{align*}
    (W \mu) (\alpha, x) \coloneqq \int_{\mathcal I} W (\alpha, \beta) \mu^\beta_t (x) \, \mathrm d\beta.
\end{align*}
To prove Theorem~\ref{thm:approxnash}, we have to establish the following result which is proven in a subsequent section. Note that this Lemma \ref{thm:muconv2} is not required in the standard graphon theory but becomes necessary for our $L^p$ graphon approach. Specifically, the proof of Lemma~\ref{lem:xconv} will include functions of the form $f_{N,i,x}'(x', \beta) = \frac{1}{\rho_N} W_N(\frac i N, \beta) \cdot \mathbf 1_{\cbc{x=x'}}$ for which a result like the one given by Theorem~\ref{thm:muconv} is required. However, the factor $\frac{1}{\rho_N}$ introduced as part of the $L^p$ graphon approach makes it impossible to bound functions such as $f_{N,i,x}'(x', \beta)$ uniformly over all $N$ because for sparse graphs $\frac{1}{\rho_N}$ approaches infinity as $N$ approaches infinity. Therefore, Theorem~\ref{thm:muconv} is not applicable in this case which makes a new statement, i.e. Lemma \ref{thm:muconv2}, necessary. Put differently, the formulation and proof of Lemma \ref{thm:muconv2} is a new mathematical contribution of our work.

\begin{lemma} \label{thm:muconv2}
	Let $\bc{c_N}_{N\geq 1}$ be a sequence of positive real numbers such that $c_N = o(1)$ and $\bc{c_N}^{-1} = o\bc{N}$ and let $B_N \subseteq \mathcal{V}_N$  be some sequence of sets with $\abs{B_N} = O \bc{N \cdot c_N}$. Then, we consider a sequence of  measurable functions $f_N \colon \mathcal X \times \mathcal I \to \mathbb R$ where for all $x \in \mathcal X$ we have $f_N (x, \alpha) = O \bc{c_N^{-1}}$ if $\alpha \in A_N \coloneqq \bigcup_{i \in B_N} \big( \frac{i -1}{N}, \frac{i}{N}\big]$ and $f_N (x, \alpha) = 0$ otherwise.
	Consider Lipschitz continuous $\boldsymbol \pi \in \boldsymbol \Pi$ up to a finite number of discontinuities $D_\pi$, with associated mean field ensemble $\boldsymbol \mu = \Psi(\boldsymbol \pi)$. Under Assumption~\ref{ass:W} and the $N$-agent policy $(\pi^1, \ldots, \pi^{i-1}, \hat \pi, \pi^{i+1}, \ldots, \pi^N) \in \Pi^N$ with $(\pi^1, \pi^2, \ldots, \pi^N) = \Gamma_N(\boldsymbol \pi) \in \Pi^N$, $\hat \pi \in \Pi$, $t \in \mathcal T$, we have
	\begin{align}
		&\E \brk{\abs{ \boldsymbol \mu^N_t \bc{f_N} - \boldsymbol \mu_t \bc{f_N}}} \to 0
	\end{align}
	uniformly over all possible deviations $\hat \pi \in \Pi, i \in  V_N$.
\end{lemma}
With Lemma~\ref{thm:muconv2} in place, we are able to provide another intermediate result which is a key element for the proof of Theorem~\ref{thm:approxnash}.
\begin{lemma} \label{lem:xconv}
Consider Lipschitz continuous $\boldsymbol \pi \in \boldsymbol \Pi$ up to a finite number of discontinuities $D_\pi$, with associated mean field ensemble $\boldsymbol \mu = \Psi(\boldsymbol \pi)$. Under Assumptions~\ref{ass:W} and \ref{ass:Lip} and the $N$-agent policy $(\pi^1, \ldots, \pi^{i-1}, \hat \pi, \pi^{i+1}, \ldots, \pi^N) \in \Pi^N$ where $(\pi^1, \pi^2, \ldots, \pi^N) = \Gamma_N(\boldsymbol \pi) \in \Pi^N$, $\hat \pi \in \Pi$ arbitrary, for any uniformly bounded family of functions $\mathcal G$ from $\mathcal X$ to $\mathbb R$ and any $\varepsilon, p > 0$, $t \in \mathcal T$, there exists $N' \in \mathbb N$ such that for all $N > N'$ we have
\begin{align} \label{eq:xconv}
    \sup_{g \in \mathcal G} \left| \E \left[ g(X^i_{t}) \right] - \E \left[ g(\hat X^{\frac i N}_{t}) \right] \right| < \varepsilon
\end{align}
uniformly over $\hat \pi \in \Pi, i \in \mathcal W_N$ for some $\mathcal W_N \subseteq V_N$ with $|\mathcal W_N| \geq \left\lfloor (1-p) N \right\rfloor$. 

Similarly, for any uniformly Lipschitz, uniformly bounded family of measurable functions $\mathcal H$ from $\mathcal X \times \mathcal B(\mathcal X)$ to $\mathbb R$ and any $\varepsilon, p > 0$, $t \in \mathcal T$, there exists $N' \in \mathbb N$ such that for all $N > N'$ we have
\begin{align} \label{eq:xmuconv}
    \sup_{h \in \mathcal H} \left| \E \left[ h(X^i_{t}, \mathbb G^{\frac i N}_N(\boldsymbol \mu^N_t)) \right] - \E \left[ h(\hat X^{\frac i N}_{t}, \mathbb G^{\frac i N}(\boldsymbol \mu_t)) \right] \right| < \varepsilon
\end{align}
uniformly over $\hat \pi \in \Pi, i \in \mathcal W_N$ for some $\mathcal W_N \subseteq V_N$ with $|\mathcal W_N| \geq \left\lfloor (1-p) N \right\rfloor$.
\end{lemma}

Finally, keeping in mind the above statements, Theorem~\ref{thm:approxnash} can be proven.

\begin{proof}[Proof of Theorem~\ref{thm:approxnash}]
Leveraging Lemma~\ref{lem:xconv}, we can establish Theorem~\ref{thm:approxnash} with an argumentation as in \cite[proofs of Corollary A.1 and Theorem 3]{cui2021learning}.
\end{proof}

\section{PROOF OF THEOREM~\ref{thm:muconv}}
\begin{proof}
We provide a proof by induction over $t$ which is structurally similar to an argument in \cite{cui2021learning}. The case $t=0$ follows from a law of large numbers argument.
For the induction step we consider the inequality
\begin{align*}
    \E \left[ \left| \boldsymbol \mu^N_{t+1}(f) - \boldsymbol \mu_{t+1}(f) \right| \right] &\leq \E \left[ \left| \boldsymbol \mu^N_{t+1}(f) - \boldsymbol \mu^N_t P^{\boldsymbol \pi^N}_{t, \boldsymbol \mu^N_t, \frac{W_N}{\rho_N}}(f) \right| \right] \\
    &\qquad \quad + \E \left[ \left| \boldsymbol \mu^N_t P^{\boldsymbol \pi^N}_{t, \boldsymbol \mu^N_t, \frac{W_N}{\rho_N}}(f) - \boldsymbol \mu^N_t P^{\boldsymbol \pi^N}_{t, \boldsymbol \mu^N_t, W}(f) \right| \right] \\
    &\qquad \quad + \E \left[ \left| \boldsymbol \mu^N_t P^{\boldsymbol \pi^N}_{t, \boldsymbol \mu^N_t, W}(f) - \boldsymbol \mu^N_t P^{\boldsymbol \pi}_{t, \boldsymbol \mu^N_t, W}(f) \right| \right] \\
    &\qquad \quad + \E \left[ \left| \boldsymbol \mu^N_t P^{\boldsymbol \pi}_{t, \boldsymbol \mu^N_t, W}(f) - \boldsymbol \mu^N_t P^{\boldsymbol \pi}_{t, \boldsymbol \mu_t, W}(f) \right| \right] \\
    &\qquad \quad + \E \left[ \left| \boldsymbol \mu^N_t P^{\boldsymbol \pi}_{t, \boldsymbol \mu_t, W}(f) - \boldsymbol \mu_{t+1}(f) \right| \right] \, .
\end{align*}
where, as before, $f\colon \mathcal X \times \mathcal I \to \mathbb R$ is an arbitrary bounded function such that $|f| \leq M_f$. While the last three terms can be bounded as in \cite{cui2021learning}, the first two summands now include the factor $\rho_N$. Thus, the first term is bounded by
\begin{align*}
    &\E \left[ \left| \boldsymbol \mu^N_{t+1}(f) - \boldsymbol \mu^N_t P^{\boldsymbol \pi^N}_{t, \boldsymbol \mu^N_t, \frac{W_N}{\rho_N}}(f) \right| \right] \\
    &\quad = \E \left[ \left| \int_{\mathcal I} \sum_{x \in \mathcal X} \mu^{N, \alpha}_{t+1}(x) \, f(x, \alpha) \, \mathrm d\alpha  -  \int_{\mathcal I} \sum_{\substack{x, x' \in \mathcal X \\ u \in \mathcal U}} \mu^{N, \alpha}_t(x)  \pi^{N,\alpha}_t(u \mid x)
    \cdot P \left(x' \mid x, u, \left(\frac{W_N}{\rho_N} \mu^N_t \right) (\alpha) \right) \, f(x', \alpha) \, \mathrm d\alpha \right| \right] \\
    &\quad = \E \left[ \left| \sum_{i \in  V_N} \left( \int_{(\frac{i-1}{N}, \frac{i}{N}]} f(X^i_{t+1}, \alpha) \, \mathrm d\alpha - \E \left[ \int_{(\frac{i-1}{N}, \frac{i}{N}]} f(X^i_{t+1}, \alpha) \, \mathrm d\alpha \innermid \mathbf X_t \right] \right) \right| \right] \\
    &\quad \leq \left( \E \left[ \left( \sum_{i \in  V_N} \left( \int_{(\frac{i-1}{N}, \frac{i}{N}]} f(X^i_{t+1}, \alpha) \, \mathrm d\alpha
    - \E \left[ \int_{(\frac{i-1}{N}, \frac{i}{N}]} f(X^i_{t+1}, \alpha) \, \mathrm d\alpha \innermid \mathbf X_t \right] \right) \right)^2 \right] \right)^{\frac 1 2} \\
    &\quad = \left( \sum_{i \in  V_N} \E \left[ \left( \int_{(\frac{i-1}{N}, \frac{i}{N}]} f(X^i_{t+1}, \alpha) \, \mathrm d\alpha
    - \E \left[ \int_{(\frac{i-1}{N}, \frac{i}{N}]} f(X^i_{t+1}, \alpha) \, \mathrm d\alpha \innermid \mathbf X_t \right] \right)^2 \right] \right)^{\frac 1 2} \leq \frac{2M_f}{\sqrt N}
\end{align*}
where we point out that the $\{ X^i_{t+1} \}_{i \in  V_N}$ are independent if conditioned on $\mathbf X_t \equiv \{ X^i_{t} \}_{i \in  V_N}$. Turning to the second summand, under Assumption~\ref{ass:W} we can derive
\begin{align*}
    &\E \left[ \left| \boldsymbol \mu^N_t P^{\boldsymbol \pi^N}_{t, \boldsymbol \mu^N_t, \frac{W_N}{\rho_N}}(f) - \boldsymbol \mu^N_t P^{\boldsymbol \pi^N}_{t, \boldsymbol \mu^N_t, W}(f) \right| \right] \\
    &\quad \leq |\mathcal X| M_f L_P \E \left[ \int_{\mathcal I} \left\Vert \int_{\mathcal I} \frac{W_N(\alpha, \beta)}{\rho_N} \mu^{N,\beta}_t \, \mathrm d\beta - \int_{\mathcal I} W(\alpha, \beta) \mu^{N,\beta}_t \, \mathrm d\beta \right\Vert \, \mathrm d\alpha \right] \\
    &\quad \leq |\mathcal X|^2 M_f L_P \sup_{x \in \mathcal X} \E \left[ \int_{\mathcal I} \left| \int_{\mathcal I} \frac{W_N(\alpha, \beta)}{\rho_N}\mu^{N,\beta}_t(x) 
    - W(\alpha, \beta)\mu^{N,\beta}_t(x) \, \mathrm d\beta \right| \, \mathrm d\alpha \right] \to 0
\end{align*}
 where $\mu^{N,\beta}_t(x)$ as a probability is less than or equal to one. The convergence in the last line is at rate $ O(1/\sqrt N)$ if the convergence rate in Assumption~\ref{ass:W} is also $ O(1/\sqrt N)$.
\end{proof}

\section{PROOF OF LEMMA~\ref{thm:muconv2}}
\begin{proof}
	The proof is by induction as follows. Note that it is similar to the proof of Theorem~\ref{thm:muconv}.
	
	\paragraph{Base case.} 
	Starting with $t=0$, we derive
	\begin{align*}
		&\E \left[ \left| \boldsymbol \mu^N_{0}\bc{f_N} - \boldsymbol \mu_{0}\bc{f_N} \right| \right] \\
		&\quad = \E \left[ \left| \int_{\mathcal I} \sum_{x \in \mathcal X} (\mu^{N, \alpha}_{0}(x) - \mu^\alpha_0(x)) f_N(x, \alpha) \, \mathrm d\alpha \right| \right] \\
		&\quad = \E \left[ \left| \int_{\mathcal I} \sum_{x \in \mathcal X} \bc{\sum_{i \in V_N} \boldsymbol{1}_{ \alpha \in (\frac{i-1}{N}, \frac{i}{N}] } \delta_{X_0^i}} \, f_N(x, \alpha) - \sum_{x \in \mathcal X} \mu^\alpha_0(x) \, f_N(x, \alpha) \, \mathrm d\alpha \right| \right] \\
		&\quad \overset{(Fubini)}{=} \E \left[ \left|\sum_{i \in  V_N}  \int_{(\frac{i-1}{N}, \frac{i}{N}]} f_N(X^i_{0}, \alpha) \, \mathrm d\alpha - \E \left[ \int_{(\frac{i-1}{N}, \frac{i}{N}]} f_N(X^i_{0}, \alpha) \, \mathrm d\alpha \right] \right| \right] \\
		&\quad \leq \left( \E \left[ \left( \sum_{i \in B_N} \left( \int_{(\frac{i-1}{N}, \frac{i}{N}]} f_N(X^i_{0}, \alpha) \, \mathrm d\alpha - \E \left[ \int_{(\frac{i-1}{N}, \frac{i}{N}]} f_N(X^i_{0}, \alpha) \, \mathrm d\alpha \right] \right) \right)^2 \right] \right)^{\frac 1 2} \\
		&\quad = \left(  \sum_{i \in B_N} \E \left[ \left( \int_{(\frac{i-1}{N}, \frac{i}{N}]} f_N(X^i_{0}, \alpha) \, \mathrm d\alpha
		- \E \left[ \int_{(\frac{i-1}{N}, \frac{i}{N}]} f_N(X^i_{0}, \alpha) \, \mathrm d\alpha \right] \right)^2 \right] \right)^{\frac 1 2}\\
		&\quad = \left(  \sum_{i \in B_N} \E \left[ \left( O \bc{c_N^{-1}} \cdot \frac{1}{N}  \right)^2 \right] \right)^{\frac 1 2}  = \left(  O \bc{N \cdot c_N} \cdot  \left( O \bc{c_N^{-2}} \cdot \frac{1}{N^2}  \right) \right)^{\frac 1 2} = \frac{O \bc{c_N^{- \frac 1 2}}}{\sqrt N} = o (1)
	\end{align*}
	by exploiting the fact that the $\cbc{X^i_0}_{i \in V_N}$ are independent and $X^i_0 \sim \mu_0 = \mu^\alpha_0$ holds $\forall i \in V_N, \alpha \in \mathcal I$.
	
	\paragraph{Induction step.} 
	For performing the induction step, we start with the following reformulation
	\begin{align*}
		\E \left[ \left| \boldsymbol \mu^N_{t+1}\bc{f_N} - \boldsymbol \mu_{t+1}\bc{f_N} \right| \right]  &\leq \E \left[ \left| \boldsymbol \mu^N_{t+1}\bc{f_N} - \boldsymbol \mu^N_t P^{\boldsymbol \pi^N}_{t, \boldsymbol \mu^N_t, \frac{W_N}{\rho_N}}\bc{f_N} \right| \right] \\
		&\quad + \E \left[ \left| \boldsymbol \mu^N_t P^{\boldsymbol \pi^N}_{t, \boldsymbol \mu^N_t, \frac{W_N}{\rho_N}}\bc{f_N} - \boldsymbol \mu^N_t P^{\boldsymbol \pi^N}_{t, \boldsymbol \mu^N_t, W}\bc{f_N} \right| \right] \\
		&\quad + \E \left[ \left| \boldsymbol \mu^N_t P^{\boldsymbol \pi^N}_{t, \boldsymbol \mu^N_t, W}\bc{f_N} - \boldsymbol \mu^N_t P^{\boldsymbol \pi}_{t, \boldsymbol \mu^N_t, W}\bc{f_N} \right| \right] \\
		&\quad + \E \left[ \left| \boldsymbol \mu^N_t P^{\boldsymbol \pi}_{t, \boldsymbol \mu^N_t, W}\bc{f_N} - \boldsymbol \mu^N_t P^{\boldsymbol \pi}_{t, \boldsymbol \mu_t, W}\bc{f_N} \right| \right] \\
		&\quad + \E \left[ \left| \boldsymbol \mu^N_t P^{\boldsymbol \pi}_{t, \boldsymbol \mu_t, W}\bc{f_N} - \boldsymbol \mu_{t+1}\bc{f_N} \right| \right]
	\end{align*}
	and assume as usual that the induction assumption is fulfilled for $t$. Apart from the second summand, all terms converge to zero by arguments as in Theorem~\ref{thm:muconv}. For the second term, however, we require a different argument. 

	Before we consider the second term itself, we make the following useful observation. By Assumption~\ref{ass:W}, we know that 
	\begin{align*} 
		&\left\Vert \frac{W_N}{\rho_N} - W \right\Vert_{L_\infty \to L_1} = \sup_{\norm{g}_\infty \leq 1} \int_{\mathcal{I}} \abs{\int_{\mathcal{I}} \left( \frac{W_N \bc{\alpha, \beta}}{\rho_N} - W \bc{\alpha, \beta}\right) g (\beta) \mathrm{d} \beta}  \mathrm{d} \alpha
	\end{align*}
	converges to zero which especially implies that for almost all $\alpha \in \mathcal{I}$ we have
	\begin{align*}
	 \abs{\int_{\mathcal{I}} \bc{\frac{W_N \bc{\alpha, \beta}}{\rho_N} - W \bc{\alpha, \beta}} \mu^{N,\beta}_t(x) \mathrm{d} \beta}  = o(1)
	\end{align*}
	for all  $x \in \mathcal{X}$ since $\mu^{N,\beta}_t(x)$ is trivially bounded by $1$. This in turn implies that for every positive real number $\epsilon > 0$ there exists a $N' \in \mathbb{N}$ such that for all $N > N'$ and $x \in \mathcal{X}$
	\begin{align}\label{int_sec_term}
		\abs{\int_{\mathcal{I}} \bc{\frac{W_N \bc{\alpha, \beta}}{\rho_N} - W \bc{\alpha, \beta}} \mu^{N,\beta}_t(x) \mathrm{d} \beta}  < \epsilon
	\end{align}
	holds for almost all $\alpha \in \mathcal{I}$. For an arbitrary but fixed $\epsilon > 0$ and under the assumption that $N > N'$, we have 
	\begin{align*}
		&\E \left[ \left| \boldsymbol \mu^N_t P^{\boldsymbol \pi^N}_{t, \boldsymbol \mu^N_t, \frac{W_N}{\rho_N}}\bc{f_N} - \boldsymbol \mu^N_t P^{\boldsymbol \pi^N}_{t, \boldsymbol \mu^N_t, W}\bc{f_N} \right| \right] \\
		&= \E \left[ \left| \int_{\mathcal I} \sum_{x \in \mathcal X} \mu^{N,\alpha}_t(x) \sum_{u \in \mathcal U} \pi^{N,\alpha}_t(u \mid x) \sum_{x' \in \mathcal X}  f_N (x', \alpha) 
		\left( P \left(x' \mid x, u, \left( \frac{W_N}{\rho_N} \mu^N_t \right) \right)  
		-  P \left(x' \mid x, u, \left( W \mu^N_t \right)  \right) \right) \mathrm d\alpha \right| \right] \\
		&= \E \left[ \left| \int_{A_N} \sum_{x \in \mathcal X} \mu^{N,\alpha}_t(x) \sum_{u \in \mathcal U} \pi^{N,\alpha}_t(u \mid x) \sum_{x' \in \mathcal X}  f_N (x', \alpha) 
        \left( P \left(x' \mid x, u, \left( \frac{W_N}{\rho_N} \mu^N_t \right) \right)  -  P \left(x' \mid x, u, \left( W \mu^N_t \right)  \right) \right) \mathrm d\alpha \right| \right] \\
		&\leq   O \bc{c_N^{-1}} \cdot  \sup_{x \in \mathcal X} \E \left[ \int_{A_N} \left| \left( \frac{W_N}{\rho_N} \mu^N_t \right) (\alpha, x) - \left( W \mu^N_t \right) (\alpha, x) \right| \, \mathrm d\alpha \right]\\
		&\overset{\mathrm{ineq.} \eqref{int_sec_term}}{\leq} O \bc{c_N^{-1}} \cdot  \sup_{x \in \mathcal X} \E \left[ \int_{A_N} \epsilon \, \mathrm d\alpha \right] 
		= O (1) \cdot \epsilon.
	\end{align*}
	Choosing $\epsilon$ arbitrarily close to zero yields the desired result, i.e. 
	\begin{align*}
		&\E \left[ \left| \boldsymbol \mu^N_t P^{\boldsymbol \pi^N}_{t, \boldsymbol \mu^N_t, \frac{W_N}{\rho_N}}\bc{f_N} - \boldsymbol \mu^N_t P^{\boldsymbol \pi^N}_{t, \boldsymbol \mu^N_t, W}\bc{f_N} \right| \right] \to 0.
	\end{align*}
	\end{proof}

\section{PROOF OF LEMMA~\ref{lem:xconv}}\label{sect:thm_2_3_last}
\begin{proof}
Parts of this proof are built on arguments in \cite[Proof of Lemma A.1]{cui2021learning}. Therefore, we focus on the differences and indicate parts that carry over from \cite{cui2021learning} to our case. For establishing \eqref{eq:xconv} we refer to \cite{cui2021learning}. Thus, it remains to show that \eqref{eq:xconv} implies \eqref{eq:xmuconv}.
As a starting point, let $\mathcal{H}$ be a  uniformly bounded (by $M_h$), uniformly Lipschitz (with Lipschitz constant $L_h$) family of measurable functions $h\colon \mathcal X \times \mathcal B(\mathcal X) \to \mathbb R$. The following reformulation holds for all $h \in \mathcal H$
\begin{align*}
    \left| \E \left[ h(X^i_{t}, \mathbb G^{\frac i N}_N(\boldsymbol \mu^N_t)) \right] - \E \left[ h(\hat X^{\frac i N}_{t}, \mathbb G^{\frac i N}(\boldsymbol \mu_t)) \right] \right|  
    &= \left| \E \left[ h(X^i_{t}, \mathbb G^{\frac i N}_N(\boldsymbol \mu^N_t)) \right] - \E \left[ h(X^i_{t}, \mathbb G^{\frac i N}_N(\boldsymbol \mu_t)) \right] \right| \\
    &\qquad + \left| \E \left[ h(X^i_{t}, \mathbb G^{\frac i N}_N(\boldsymbol \mu_t)) \right] - \E \left[ h(X^i_{t}, \mathbb G^{\frac i N}(\boldsymbol \mu_t)) \right] \right| \\
    &\qquad + \left| \E \left[ h(X^i_{t}, \mathbb G^{\frac i N}(\boldsymbol \mu_t)) \right] - \E \left[ h(\hat X^{\frac i N}_{t}, \mathbb G^{\frac i N}(\boldsymbol \mu_t)) \right] \right|.
\end{align*}
Focusing on the first term, we obtain
\begin{align*}
    \left| \E \left[ h(X^i_{t}, \mathbb G^{\frac i N}_N(\boldsymbol \mu^N_t)) \right] - \E \left[ h(X^i_{t}, \mathbb G^{\frac i N}_N(\boldsymbol \mu_t)) \right] \right|
    & \leq \E \left[ \E \left[ \left| h(X^i_{t}, \mathbb G^{\frac i N}_N(\boldsymbol \mu^N_t)) h(X^i_{t}, \mathbb G^{\frac i N}_N(\boldsymbol \mu_t)) \right| \innermid X^i_{t} \right] \right] \\
    & \leq L_h \E \left[ \left\Vert \mathbb G^{\frac i N}_N(\boldsymbol \mu^N_t) - \mathbb G^{\frac i N}_N(\boldsymbol \mu_t) \right\Vert \right] \\
    & = L_h \sum_{x' \in \mathcal X} \E \left[ \left| \int_{\mathcal I} \frac{1}{\rho_N} W_N(\frac i N, \beta) ( \mu^{N,\beta}_t(x) - \mu^\beta_t(x))  \mathrm d\beta \right| \right]
\end{align*}
which converges to zero. To see this, define the functions $f_{N,i,x}'(x', \beta) = \frac{1}{\rho_N} W_N(\frac i N, \beta) \cdot \mathbf 1_{\cbc{x=x'}}$ and apply Lemma~\ref{thm:muconv2} to them. Combining our findings for the first term with results in \cite{cui2021learning} for the second and third term concludes the proof.
\end{proof}

\section{PROOF OF THEOREM~\ref{thm:smoothed_reward}}\label{sec:smoothed_proofs}

From now on, we use the notation $\widehat{W}$ for the smoothed step graphon and $W$ is the original step graphon. Ultimately, our goal is to establish Theorem~\ref{thm:smoothed_reward}. Therefore, we have to provide some intermediate results first for which the corresponding proofs can be found in the subsequent sections. Note that the structure and proofs of the intermediate results follow ideas in \cite{cui2021learning} although the referenced paper does not consider smoothed step graphons.
\begin{lemma} \label{lem:smoothed}
Let $\boldsymbol \pi \in \boldsymbol \Pi$ be a policy ensemble generating  $\hat{\boldsymbol{\mu}} \in \boldsymbol{\mathcal{M}}$ in the smoothed MP MFG and generating $\boldsymbol{\mu} \in \boldsymbol{\mathcal{M}}$ in the standard MP MFG under Assumption~\ref{ass:Lip}. Then, we have for all $t \in \mathcal{T}$ and all measurable functions $f: \mathcal{X} \times \mathcal{I} \to \mathbb{R}$ uniformly bounded by some $M_f > 0$ that
\begin{align}
     \left\vert \hat{\boldsymbol{\mu}}_t (f) - \boldsymbol{\mu}_t (f) \right\vert \to 0 \quad \textrm{as} \quad \xi \to 0 
\end{align}
uniformly over all $\boldsymbol \pi \in \boldsymbol \Pi$.
\end{lemma}
This first lemma in turn allows us to provide the following result.
\begin{lemma}\label{lem:smoothed_state}
In the setup of Lemma~\ref{lem:smoothed}, for any uniformly bounded family of functions $\mathcal{G}$  from $\mathcal{X}$ to $\mathbb{R}$ and any $\varepsilon, p > 0$, $t \in \mathcal{T}$, there exists $\xi' > 0$ such that for all $\xi \in (0, \xi')$
\begin{align}\label{eq:smoothed_state}
    \sup_{g \in \mathcal{G}} \left\vert \Erw \left[g (X_t^\alpha) \right] - \Erw \left[g (\hat{X}_t^\alpha) \right] \right\vert < \varepsilon
\end{align}
holds for all $\alpha \in \mathcal{J} \subseteq \mathcal{I}$ where $\lambda (\mathcal{J}) \geq 1-p$. Under the same conditions, we have that for any uniformly bounded, uniformly Lipschitz family of measurable functions $\mathcal{H}$ from $\mathcal{X} \times \mathcal{B}_1 (\mathcal{X})$ to $\mathbb{R}$ and any $\varepsilon, p > 0$, $t \in \mathcal{T}$, there exists $\xi' > 0$ such that for all $\xi \in (0, \xi')$
\begin{align}\label{eq:smoothed_mf_and_state}
    \sup_{h \in \mathcal{H}} \left\vert \Erw \left[h (X_t^\alpha, \mathbb{G}_W^{\alpha} (\boldsymbol{\mu}_t)) \right] - \Erw \left[h (\hat{X}_t^\alpha), \mathbb{G}_{\widehat{W}}^{\alpha} (\hat{\boldsymbol{\mu}}_t) \right] \right\vert < \varepsilon
\end{align}
holds for all $\alpha \in \mathcal{J} \subseteq \mathcal{I}$ where $\lambda (\mathcal{J}) \geq 1-p$.
\end{lemma}

The next statement is a consequence of the just stated results.
\begin{corollary}\label{cor:smoothed_reward}
In the setup of Lemma~\ref{lem:smoothed}, for every $\varepsilon, p > 0$ there exists a $\xi' > 0$ such that for all $\xi \in (0, \xi')$ 
\begin{align*}
    \sup_{\pi \in \Pi} \vert J^t_{\alpha, W} (\pi) - J^t_{\alpha, \widehat{W}} (\pi) \vert < \varepsilon \, .
\end{align*}
holds for all $\alpha \in \mathcal{J} \subseteq \mathcal{I}$ where $\lambda (\mathcal{J}) \geq 1-p$.
\end{corollary}
With these statements in place, we can prove the theorem of interest.

\begin{proof}[Proof of Theorem~\ref{thm:smoothed_reward}]
We consider the inequality
\begin{align*}
    \sup_{\boldsymbol{\pi} \in \Pi}  J_{\alpha, W}(\boldsymbol{\pi}_{\mathrm{smo}}) - J_{\alpha, W} (\boldsymbol{\pi}) 
    &\leq \sup_{\boldsymbol{\pi} \in \Pi} \left\vert J_{\alpha, W}(\boldsymbol{\pi}_{\mathrm{smo}})- J_{\alpha, W_{\mathrm{smo}}}(\boldsymbol{\pi}_{\mathrm{smo}}) \right\vert \\
    &\quad + \sup_{\boldsymbol{\pi} \in \Pi} J_{\alpha, W_{\mathrm{smo}}}(\boldsymbol{\pi}_{\mathrm{smo}}) - J_{\alpha, W_{\mathrm{smo}}}(\boldsymbol{\pi})  \\
    &\quad + \sup_{\boldsymbol{\pi} \in \Pi} \left\vert J_{\alpha, W_{\mathrm{smo}}}(\boldsymbol{\pi}) - J_{\alpha, W} (\boldsymbol{\pi}) \right\vert \, .
\end{align*}
The second term is equal to zero because $(\boldsymbol{\pi}_{\mathrm{smo}}, \boldsymbol{\mu}_{\mathrm{smo}})$ is a GMFE for the smoothed game. The first and third term, on the other hand, can be bounded by $\epsilon /2$ each by using Corollary~\ref{cor:smoothed_reward} for all $\alpha \in \mathcal{J}$ for some $\mathcal{J} \subseteq \mathcal{I}$ with Lebesgue measure $\lambda (\mathcal{J}) \geq 1-p$. Eventually, this yields
\begin{align*}
    &\sup_{\boldsymbol{\pi} \in \Pi} \left\vert J_{\alpha, W} (\boldsymbol{\pi}) - J_{\alpha, W}(\boldsymbol{\pi}_{\mathrm{smo}}) \right\vert  \leq \frac{\varepsilon}{2} + 0 + \frac{\varepsilon}{2} = \varepsilon
\end{align*}
and thereby concludes the proof.
\end{proof}

\section{PROOF OF LEMMA~\ref{lem:smoothed}}
We prove the statement via induction over $t$. The base case trivially holds. For the induction step, we consider the following reformulation
\begin{align*}
    \left\vert \hat{\boldsymbol{\mu}}_{t+1} (f) - \boldsymbol{\mu}_{t+1} (f) \right\vert 
    &\leq  \left| \hat{\boldsymbol{\mu}}_t  P^{\boldsymbol \pi}_{t,  \hat{\boldsymbol{\mu}}_t, \widehat{W}} (f) - \hat{\boldsymbol{\mu}}_t P^{\boldsymbol \pi}_{t, \hat{\boldsymbol{\mu}}_t, W}(f) \right|  \\
    &\qquad \quad +  \left| \hat{\boldsymbol{\mu}}_t P^{\boldsymbol \pi}_{t, \hat{\boldsymbol{\mu}}_t, W}(f) - \hat{\boldsymbol{\mu}}_t P^{\boldsymbol \pi}_{t, \boldsymbol{\mu}_t, W}(f) \right| \\
    &\qquad \quad +  \left| \hat{\boldsymbol{\mu}}_t P^{\boldsymbol \pi}_{t, \boldsymbol{\mu}_t, W}(f) - \boldsymbol \mu_{t+1}(f) \right| \, .
\end{align*}
We consider the three summands separately.

\paragraph{First term.}
For the first term we note that
\begin{align*}
    \left| \hat{\boldsymbol{\mu}}_t  P^{\boldsymbol \pi}_{t,  \hat{\boldsymbol{\mu}}_t, \widehat{W}} (f) - \hat{\boldsymbol{\mu}}_t P^{\boldsymbol \pi}_{t, \hat{\boldsymbol{\mu}}_t, W}(f) \right| 
    &\leq |\mathcal X| M_f L_P \int_{\mathcal I} \left\Vert \int_{\mathcal I} \widehat{W}(\alpha, \beta) \hat{\mu}^\beta_t \, \mathrm d\beta  - \int_{\mathcal I} W(\alpha, \beta) \hat{\mu}^\beta_t \, \mathrm d\beta \right\Vert \, \mathrm d\alpha \\
    & \leq |\mathcal X|^2 M_f L_P \sup_{x \in \mathcal X} \int_{\mathcal I} \left| \int_{\mathcal I} \widehat{W}(\alpha, \beta) \hat{\mu}^\beta_t(x)  -  W(\alpha, \beta)\hat{\mu}^\beta_t(x) \, \mathrm d\beta \right| \, \mathrm d\alpha
\end{align*}
converges to zero as $\xi \to 0$ by the construction of the smoothed step graphon $\widehat{W}$.
\paragraph{Second term.} Consider the following reformulation
\begin{align*}
    &\left| \hat{\boldsymbol{\mu}}_t P^{\boldsymbol \pi}_{t, \hat{\boldsymbol{\mu}}_t, W}(f) - \hat{\boldsymbol{\mu}}_t P^{\boldsymbol \pi}_{t, \boldsymbol{\mu}_t, W}(f) \right| \\
    &\quad \leq M_f |\mathcal X|  \sup_{x,u} \int_{\mathcal I} \left| P \left(x' \mid x, u, \int_{\mathcal I} W(\alpha, \beta) \hat{\mu}^\beta_t \, \mathrm d\beta \right) -  P \left(x' \mid x, u, \int_{\mathcal I} W(\alpha, \beta) \mu^\beta_t \, \mathrm d\beta \right) \right| \, \mathrm d\alpha  \\
    &\quad \leq M_f |\mathcal X| L_P \sum_{x' \in \mathcal X} \int_{\mathcal I}   \left| \int_{\mathcal I} W(\alpha, \beta) \hat{\mu}^\beta_t(x')\, \mathrm d\beta - \int_{\mathcal I} W(\alpha, \beta) \mu^\beta_t(x') \, \mathrm d\beta \right| \, \mathrm d\alpha.
\end{align*}
Then, the induction assumption can be applied to
\begin{align*}
    &\int_{\mathcal I}  \left| \int_{\mathcal I} W(\alpha, \beta) \hat{\mu}^\beta_t(x')\, \mathrm d\beta - \int_{\mathcal I} W(\alpha, \beta) \mu^\beta_t(x') \, \mathrm d\beta \right| \, \mathrm d\alpha  = \int_{\mathcal I} \left| \boldsymbol \hat{\mu}_t(f'_{\alpha, x'}) - \boldsymbol \mu_t(f'_{\alpha, x'}) \right| \, \mathrm d\alpha \to 0
\end{align*}
for $\xi \to 0$ where we choose $f'_{x', \alpha}(x, \beta) = W(\alpha, \beta) \cdot \mathbf 1_{\cbc{x=x'}}$. This immediately implies
\begin{align*}
    & \left| \hat{\boldsymbol{\mu}}_t P^{\boldsymbol \pi}_{t, \hat{\boldsymbol{\mu}}_t, W}(f) - \hat{\boldsymbol{\mu}}_t P^{\boldsymbol \pi}_{t, \boldsymbol{\mu}_t, W}(f) \right| \to 0 \quad \textrm{as} \quad \xi \to 0.
\end{align*}

\paragraph{Third term.}
For the last term we have by the induction assumption
\begin{align*}
     \left| \hat{\boldsymbol{\mu}}_t P^{\boldsymbol \pi}_{t, \boldsymbol{\mu}_t, W}(f) - \boldsymbol \mu_{t+1}(f) \right|
    &= \left| \int_{\mathcal I} \sum_{x \in \mathcal X} \hat{\mu}^\alpha_t(x) f''(x, \alpha) \, \mathrm d\alpha - \int_{\mathcal I} \sum_{x \in \mathcal X} \mu^\alpha_t(x) f''(x, \alpha) \, \mathrm d\alpha \right|\\
    &\quad = \left| \hat{\boldsymbol{\mu}}_t(f'') - \boldsymbol \mu_t(f'') \right| \to 0 \quad \textrm{as} \, \, \xi \to 0.
\end{align*}
Here we have applied the induction assumption to the function $f''$ defined by
\begin{align*}
    &f''(x, \alpha) = \sum_{\substack{u \in \mathcal U \\ x' \in \mathcal X}}\pi^\alpha_t(u \mid x) P \left(x' \mid x, u, \int_{\mathcal I} W(\alpha, \beta) \mu^{\beta}_t \, \mathrm d\beta \right) \, f(x', \alpha).
\end{align*}

\section{PROOF OF LEMMA~\ref{lem:smoothed_state}}

We start by showing that \eqref{eq:smoothed_state} implies \eqref{eq:smoothed_mf_and_state}. Subsequently, we establish \eqref{eq:smoothed_state} to complete the proof. First, note that \begin{align*}
    \left\vert \Erw \left[h (X_t^\alpha, \mathbb{G}_W^{\alpha} (\boldsymbol{\mu}_t)) \right] - \Erw \left[h (\hat{X}_t^\alpha, \mathbb{G}_{\widehat{W}}^\alpha (\hat{\boldsymbol{\mu}}_t)) \right] \right\vert
    &= \left\vert \Erw \left[h (X_t^\alpha, \mathbb{G}_W^{\alpha} (\boldsymbol{\mu}_t)) \right] - \Erw \left[h (X_t^\alpha, \mathbb{G}_{\widehat{W}}^\alpha (\boldsymbol{\mu}_t)) \right] \right\vert \\
    &+ \left\vert \Erw \left[h (X_t^\alpha, \mathbb{G}_{\widehat{W}}^{\alpha} (\boldsymbol{\mu}_t)) \right] - \Erw \left[h (X_t^\alpha, \mathbb{G}_{\widehat{W}}^{\alpha} (\hat{\boldsymbol{\mu}}_t)) \right] \right\vert \\
    &+ \left\vert \Erw \left[h (X_t^\alpha, \mathbb{G}_{\widehat{W}}^{\alpha} (\hat{\boldsymbol{\mu}}_t)) \right] - \Erw \left[h (\hat{X}_t^\alpha, \mathbb{G}_{\widehat{W}}^{\alpha} (\hat{\boldsymbol{\mu}}_t)) \right] \right\vert \, .
\end{align*}
Let us consider the three terms separately. We denote by $L_h$ the Lipschtiz constant and by $M_h$ the uniform bound of $\mathcal{H}$.
\paragraph{First term.}
For the first term we have
\begin{align*}
    \left\vert \Erw \left[h (X_t^\alpha, \mathbb{G}_W^{\alpha} (\boldsymbol{\mu}_t)) \right] - \Erw \left[h (X_t^\alpha, \mathbb{G}_{\widehat{W}}^{\alpha} (\boldsymbol{\mu}_t)) \right] \right\vert
    &\leq L_h \Erw \left[ \left\Vert \mathbb{G}_W^{\alpha} (\boldsymbol{\mu}_t)) -  \mathbb{G}_{\widehat{W}}^{\alpha} (\boldsymbol{\mu}_t)) \right\Vert \right] \\
    &\quad = L_h \sum_{x \in \mathcal{X}} \left\vert \int_{\mathcal{I}} \left( W (\alpha, \beta) - \widehat{W} (\alpha, \beta) \right) \mu_t^\beta (x) \mathrm{d} \beta \right\vert
\end{align*}
which converges to $0$ as $\xi \to 0$ for arbitrary fractions of the possible $\alpha$ since the set of points $(\alpha, \beta) \in \mathcal{I}^2$ where $W$ and $\widehat{W}$ differ becomes arbitrarily small by construction for $\xi$ close enough to $0$.

\paragraph{Second term.}
Using a similar reformulation as for the first term, we obtain
\begin{align*}
    &\left\vert \Erw \left[h (X_t^\alpha, \mathbb{G}_{\widehat{W}}^{\alpha} (\boldsymbol{\mu}_t)) \right] - \Erw \left[h (X_t^\alpha, \mathbb{G}_{\widehat{W}}^{\alpha} (\hat{\boldsymbol{\mu}}_t)) \right] \right\vert\leq L_h \sum_{x \in \mathcal{X}} \left\vert \int_{\mathcal{I}} \left( \mu_t^\beta (x) - \hat{\mu}_t^\beta (x) \right) \widehat{W} (\alpha, \beta)  \mathrm{d} \beta \right\vert
\end{align*}
which converges to zero as $\xi \to 0$. This follows from applying Lemma~\ref{lem:smoothed} to the functions $f_{\alpha, x} (x', \beta) = \widehat{W} (\alpha, \beta) \boldsymbol{1}_{\{ x = x' \}}$.

\paragraph{Third term.} Finally, the last term
\begin{align*}
 \left\vert \Erw \left[h (X_t^\alpha, \mathbb{G}_{\widehat{W}}^{\alpha} (\hat{\boldsymbol{\mu}}_t)) \right] - \Erw \left[h (\hat{X}_t^\alpha, \mathbb{G}_{\widehat{W}}^{\alpha} (\hat{\boldsymbol{\mu}}_t)) \right] \right\vert
\end{align*}
converges to zero as $\xi \to 0$ by \eqref{eq:smoothed_state}. Thus, we have established that \eqref{eq:smoothed_state} implies \eqref{eq:smoothed_mf_and_state}.

Now it remains to show \eqref{eq:smoothed_state}. We prove \eqref{eq:smoothed_state} via induction over $t$. For $t=0$, the statement trivially holds. For the induction step, we note that
\begin{align*}
    &\left\vert \Erw \left[g (X_{t+1}^\alpha) \right] - \Erw \left[g (\hat{X}_{t+1}^\alpha) \right] \right\vert = \left\vert \Erw \left[f'_t (X_t^\alpha, \mathbb{G}^\alpha_W (\boldsymbol{\mu}_t)) \right] - \Erw \left[f'_t (\hat{X}_t^\alpha, \mathbb{G}^\alpha_{\widehat{W}} (\hat{\boldsymbol{\mu}}_t)) \right] \right\vert
\end{align*}
with
\begin{align*}
    f'_t (x, \nu) = \sum_{u \in \mathcal{U}} \pi_t (u \mid x) \sum_{x' \in \mathcal{X}} P (x' \mid x, u, \nu) g(x').
\end{align*}
Then we can apply \eqref{eq:smoothed_mf_and_state} to 
\begin{align*}
   \left\vert \Erw \left[f'_t (X_t^\alpha, \mathbb{G}^\alpha_W (\boldsymbol{\mu}_t)) \right] - \Erw \left[f'_t (\hat{X}_t^\alpha, \mathbb{G}^\alpha_{\widehat{W}} (\hat{\boldsymbol{\mu}}_t)) \right] \right\vert
\end{align*}
to obtain the desired result. Note that \eqref{eq:smoothed_mf_and_state} holds by the induction assumption because we have already established that \eqref{eq:smoothed_state} implies \eqref{eq:smoothed_mf_and_state}. This concludes the proof.

\section{PROOF OF COROLLARY~\ref{cor:smoothed_reward}}\label{sec:thm_smoothed_last}
The statement follows from Lemma~\ref{lem:smoothed_state}. To see this, we define the functions
\begin{align*}
    r_t (x, \nu) = \sum_{u \in \mathcal{U}} \pi_t (u \mid x) r(x, u, \nu)
\end{align*}
which are uniformly Lipschitz and uniformly bounded by construction. Then, we obtain by Lemma~\ref{lem:smoothed_state}
\begin{align*}
    &\vert J^t_{\alpha, W}(\pi) - J^t_{\alpha, \widehat{W}}(\pi)\vert \leq \sum_{t =0}^{T-1} \left\vert \Erw \left[r_t (X_t^\alpha, \mathbb{G}^\alpha_W (\boldsymbol{\mu}_t)) \right] - \Erw \left[r_t (\hat{X}_t^\alpha, \mathbb{G}^\alpha_{\widehat{W}} (\hat{\boldsymbol{\mu}}_t)) \right] \right\vert \leq \varepsilon
\end{align*}
for all $\alpha \in \mathcal{J}$ for some $\mathcal{J} \subseteq \mathcal{I}$ with $\lambda (\mathcal{J}) \geq 1-p$. This concludes the proof.

\section{DETAILS ON LEARNING LPGMFGS}\label{app:learning_methods}

Before we prove the theoretical results on learning LPGMFGs, we introduce some terminology and technical details that were not included in the main text for space reasons. Our implementation is based on \cite{cui2021learning} (MIT license). For Figure~\ref{fig:network-compare}, we used data from \cite{rozemberczki2019gemsec} (GPL-3.0 license). The dataset from \cite{rozemberczki2019gemsec} contains mutual likes among verified Facebook pages of TV shows that were obtained from public APIs. Thus, the data set does not contain any personal information of individuals and no person is identifiable. Our code is available online\footnote{Link: \url{https://github.com/ChrFabian/Learning_sparse_GMFGs}}. The total amount of compute to reproduce our work amounts to approximately 2 days on a 64 core AMD Epyc processor and 64 GB of RAM. No GPUs were used in this research.

\paragraph{Equivalence Classes.}
For learning LPGMFGs, we apply the well-known equivalence class method, see for example \cite{cui2021learning}.  Therefore, we discretize the continuous spectrum of agents $\mathcal{I}$ into $M + 1$ disjoint subsets.  Formally, we consider a grid of agents $\bc{\alpha_m}_{0\leq m \leq M}$ with $\alpha_m = m /100$ where each $\alpha_m$ is associated with a policy $\pi^{\alpha_m}\in \mathcal{P} \bc{\mathcal{U}}^{\mathcal{T}}$ and a mean field $\mu^{\alpha_m} \in \mathcal{P} \bc{\mathcal{X}}^{\mathcal{T}}$. Then, each agent $\alpha \in \brk{0,1}$ follows the policy $\pi^{\alpha_m}$ of the agent at grid point $\alpha_m$ closest to $\alpha$. Thus, there are $M+ 1$ intervals such that $\mathcal{I} = \bigcup_{m =0}^M \tilde{\mathcal{I}}_m$ with $\tilde{\mathcal{I}}_m = \big[ \frac{m - 1/2}{M}, \frac{m+1/2}{M}\big)$ for $0< m < M$,$\tilde{\mathcal{I}}_0 = \big[ 0, \frac{1}{2 M} \big]$, and $\tilde{\mathcal{I}}_M = \big[ \frac{M- 1}{M}, 1 \big]$. Note that all agents in an interval follow the same policy. For any fixed policy ensemble $\boldsymbol{\pi} \in \boldsymbol{\Pi}$ we define
\begin{align*}
\hat{\Psi}(\boldsymbol \pi) = \sum_{m= 1}^M \boldsymbol{1}_{\cbc{\alpha \in \tilde{\mathcal{I}}_m}} \hat{\mu}^{\alpha_m}
\end{align*}
where $\hat{\boldsymbol{\mu}}$ is recursively defined by
\begin{align*}
    \hat{\mu}_{t + 1}^{\alpha_m} (x) = \sum_{\substack{x' \in \mathcal X \\ u \in \mathcal{U}}} \hat{\mu}_{t}^{\alpha_m} \bc{x'} \pi_t^{\alpha_m} \bc{u \vert x'} P \bc{x \vert x', u, \hat{\mathbb{G}}_t^{\alpha_m}}
\end{align*}
with $\hat{\mu}_{0}^{\alpha_m} (x)  = \mu_{0} (x)$ for $x \in \mathcal{X}$ and $m = 1, \ldots, M$, and neighborhood mean fields 
\begin{align*}
	\hat{\mathbb{G}}_t^{\alpha} = \frac{1}{M}\sum_{m= 1}^M W \bc{\alpha, \alpha_m}  \hat{\mu}_{t}^{\alpha_m}.
\end{align*}
Similarly, we approximate the policy ensemble by
\begin{align*}
	\hat{\Phi} \bc{\boldsymbol{\mu}} = \sum_{m= 1}^M \boldsymbol{1}_{\cbc{\alpha \in \tilde{\mathcal{I}}_m}} \pi^{\alpha_m}
\end{align*}
where $\pi^{\alpha_m}$ is the optimal policy of $\alpha_i$ for any fixed $\boldsymbol{\mu}$.

\paragraph{Multi-Population Mean Field Games.}
In this paragraph, we briefly introduce some crucial definitions for the following proofs. For more details on MP MFGs, see for example \cite{perolat2021scaling}. A player of population $i$ who implements policy $\pi^i$ while the mean field of players is distributed according to $\boldsymbol{\mu} \in \boldsymbol{\mathcal M}$ can expect the following sum of rewards
\begin{align*}
	J_i^{\boldsymbol{\mu}} \bc{\pi^i} &= \Erw \left[\sum_{t =0}^T r \bc{x^i_t, u^i_t, \mu_t} \Bigg\vert x^i_0 \sim \mu_0^i, u^i_t \sim \pi_t^i \bc{\cdot \big\vert x^i_t},  x_{t+1}^i \sim P \bc{\cdot \big\vert x_t^i, u_t^i, \mu_t} \right].
\end{align*}
Similarly, we define the corresponding $Q$-function
\begin{align*}
	&Q_i^{\boldsymbol{\mu}, \pi^i} \bc{t, x^i, u^i} = \Erw \left[ \sum_{k =t}^T r \bc{x^i_k, u^i_k, \mu_k} \Bigg\vert x_t^i = x^i, u_t^i = u^i,  u^i_k \sim \pi_t^i \bc{\cdot \big\vert x^i_k}, x_{k+1}^i \sim P \bc{\cdot \big\vert x_k^i, u_k^i, \mu_k} \right]
\end{align*}
and the corresponding value function
\begin{align*}
&V_i^{\boldsymbol{\mu}, \pi^i} \bc{t, x^i} = \Erw \left[ \sum_{k =t}^T r \bc{x^i_k, u^i_k, \mu_k} \Bigg\vert x_t^i = x^i, u^i_k \sim \pi_t^i \bc{\cdot \big\vert x^i_k}, x_{k+1}^i \sim P \bc{\cdot \big\vert x_k^i, u_k^i, \mu_k} \right].
\end{align*}
The sequence of state distributions $\boldsymbol{\mu}^{\boldsymbol{\pi}, \boldsymbol{\mu}'} \in \boldsymbol{\mathcal M}$ is defined by the forward equation
\begin{align*}
	&\mu_{t+1}^{i, \boldsymbol{\pi}, \boldsymbol{\mu}'} \bc{{x'}^i} \coloneqq \sum_{\bc{x^i, u^i} \in \mathcal{X} \times \mathcal{U}} \pi_t^i \bc{u^i \vert x^i} \mu_{t}^{i, \boldsymbol{\pi}, \boldsymbol{\mu}'} \bc{{x}^i} P \bc{{x'}^i \big\vert x^i, u^i, \mu'_t}
\end{align*}
for $(\boldsymbol \pi, \boldsymbol \mu) \in \boldsymbol \Pi \times \boldsymbol{\mathcal M}$. In the case of $\boldsymbol{\mu} = \boldsymbol{\mu}'$, we will often use the abbreviated notation $\boldsymbol{\mu}^{\boldsymbol{\pi}} = \boldsymbol{\mu}^{\boldsymbol{\pi}, \boldsymbol{\mu}'}$.

\paragraph{OMD Algorithm.}
As we will see in the subsequent statements and proofs, there is some terminology necessary for analyzing OMDs such as the concept of a regularizer $h\colon \mathcal{P} (\mathcal{U}) \to \mathbb{R}$. Here, we choose the entropy as the regularizer, i.e. $h(\pi) \coloneqq - \sum_{u \in \mathcal{U}} \pi(u) \log(\pi(u))$.
Then, the convex conjugate of $h$ is defined as $h^* (y) = \max_{p \in \mathcal{P}(\mathcal{U})} \bc{\scal{y}{p} - h (\pi)} = \log (\sum_{u \in \mathcal{U}} \exp(y(u)))$. Finally, we define for (almost) every $y$
\begin{align*}
\Gamma (y) \coloneqq \nabla h^* (y) &= \argmax_{p \in \mathcal{P}(\mathcal{U})} \bc{\scal{y}{p} - h (\pi)} 
= \frac{\exp (y )}{\sum_{u' \in \mathcal{U}} \exp(y (u'))}.
\end{align*}

A more detailed introduction to these concepts can be found, for example, in \cite{perolat2021scaling}. For our learning approach, we apply the OMD algorithm as formulated in \cite{perolat2021scaling}, see Algorithm~\ref{alg:discrete}.

\paragraph{Proof Strategy} The theoretical analysis is conducted by considering the continuous time analogue (\cite{mertikopoulos2018cycles}) of Algorithm~\ref{alg:discrete} on the smoothed step graphon, which we refer to as smoothed CTOMD, and which is characterized by the equations 
\begin{align}\label{eq: CTOMD_dynamics_y}
    y_{t, \tau}^\alpha \bc{x^\alpha, u^\alpha} = \int_0^\tau Q_\alpha^{\pi_s^\alpha, \mu^{\pi_s}} \bc{x^\alpha, u^\alpha} \textrm{d} s 
\end{align} 
and
\begin{align}\label{eq: CTOMD_dynamics_pi}
    \pi_{t, \tau}^\alpha (\cdot \vert x^\alpha) = \Gamma \bc{y_{t, \tau}^\alpha (x^\alpha, \cdot)}
\end{align}
for all $\alpha \in \mathcal{I}$, $t \in \mathcal{T}$, and $\tau \in \mathbb{R}_+$ with initial values $y_{t, 0}^i = 0$. Taking the derivative with respect to $\tau$ on both sides of equation \eqref{eq: CTOMD_dynamics_y} yields
\begin{align*}
	\frac{\mathrm d}{\mathrm d \tau}y_{t, \tau}^\alpha \bc{x^\alpha, u^\alpha} = Q_\alpha^{\pi_\tau^\alpha, \mu^{\pi_\tau}} \bc{x^\alpha, u^\alpha}
\end{align*}
and taking the derivative with respect to $\tau$ on both sides of equation \eqref{eq: CTOMD_dynamics_pi} results in
\begin{align*}
	\frac{\mathrm d}{\mathrm d \tau} \pi_{t, \tau}^\alpha (u^\alpha \vert x^\alpha)
	&= \frac{\mathrm d}{\mathrm d \tau} \Gamma \bc{y_{t, \tau}^\alpha (x^\alpha, u^\alpha)} \\
	&= \exp (y_{t, \tau}^\alpha (x^\alpha, u^\alpha)) \cdot Q_\alpha^{\pi_\tau^\alpha, \mu^{\pi_\tau}} \bc{x^\alpha, u^\alpha} \cdot \frac{\sum_{u' \in \mathcal{U}} \exp (y_{t, \tau}^\alpha (x^\alpha, {u'}^\alpha))}{\bc{\sum_{u' \in \mathcal{U}} \exp (y_{t, \tau}^\alpha (x^\alpha, {u'}^\alpha))}^2}\\
	&\quad - \brk{\sum_{u' \in \mathcal{U}} Q_\alpha^{\pi_\tau^\alpha, \mu^{\pi_\tau}} \bc{x^\alpha, {u'}^\alpha} \exp (y_{t, \tau}^\alpha (x^\alpha, {u'}^\alpha))} \cdot \frac{\exp (y_{t, \tau}^\alpha (x^\alpha, u^\alpha))}{\bc{\sum_{u' \in \mathcal{U}} \exp (y_{t, \tau}^\alpha (x^\alpha, {u'}^\alpha))}^2}\\
	&= \pi_{t, \tau}^\alpha (u^\alpha \vert x^\alpha) \cdot Q_\alpha^{\pi_\tau^\alpha, \mu^{\pi_\tau}} - \pi_{t, \tau}^\alpha (u^\alpha \vert x^\alpha) \cdot \sum_{u' \in \mathcal{U}} Q_\alpha^{\pi_\tau^\alpha, \mu^{\pi_\tau}} \bc{x^\alpha, {u'}^\alpha} \pi_{t, \tau}^\alpha ({u'}^\alpha \vert x^\alpha) 
	\eqqcolon g ( \pi_{t, \tau}^\alpha (u^\alpha \vert x^\alpha)).
\end{align*}
Following the above derivations, our goal is to solve the differential equation
\begin{align}\label{eq: CTOMD_diff_eq}
	\frac{\mathrm d}{\mathrm d \tau} \pi_{t, \tau}^\alpha (u^\alpha \vert x^\alpha) = g (\pi_{t, \tau}^\alpha (u^\alpha \vert x^\alpha)).
\end{align}

The strategy of this section is to eventually apply LaSalle's theorem \cite[Theorem 4.4]{khalil2002nonlinear} to equation \eqref{eq: CTOMD_diff_eq} which will yield Theorem~\ref{thm:unique_conv}. First, we define a measure of similarity
\begin{align}\label{def:sim_measure}
	H(y) \coloneqq &\int_{\mathcal{I}} \sum_{t = 0}^T \sum_{x^i \in \mathcal{X}} \mu_t^{i, \boldsymbol{\pi}^*} \bc{x^i}  \left[ h^* \bc{y_{t, \tau}^i \bc{x^i, \cdot}} - h^* \bc{y_t^{i, *}\bc{x^i, \cdot}} -\scal{\pi_t^{i, *} \bc{\cdot \vert x^i}}{y_{t, \tau}^i \bc{x^i, \cdot} - y_t^{i, *}\bc{x^i, \cdot}} \right] \mathrm d \alpha 
\end{align}
using the convention that $\Gamma (y_t^{i, *}\bc{x^i, \cdot}) = \pi_t^{i, *} \bc{\cdot \vert x^i}$, where $\boldsymbol{\pi}^*$ is the unique NE.
Next, we make a rather technical observation.
\begin{lemma}\label{lem:derivative_con_conj}
The convex conjugate of the regularizer satisfies
\begin{align*}
	\frac{\mathrm d}{\mathrm d \tau}h^* \bc{y_{t, \tau}^\alpha \bc{x^\alpha, \cdot}} = \scal{\pi_{t, \tau}^\alpha \bc{\cdot \big\vert x^\alpha}}{Q_\alpha^{\boldsymbol{\mu}^{\boldsymbol{\pi}_\tau}, \pi^\alpha_\tau} \bc{t, x^\alpha, \cdot}}.
\end{align*}
\end{lemma}

\begin{algorithm}[t]
    \caption{\textbf{Online Mirror Descent (\cite{perolat2021scaling})}}
    \label{alg:discrete}
    \begin{algorithmic}[1]
        \STATE \textbf{Input}: $\gamma$, number of iterations $\tau_{\max}$, and $\forall i \leq M, t \in \mathcal{T}: y_{t, 0}^i = 0$ 
        \FOR {$\tau=1, \ldots, \tau_{\max}$}
            \STATE Forward update for all $i$: $\mu^{i, \pi_\tau}$
            \STATE Backward update for all $i$: $Q_i^{\pi_\tau^i, \mu^{\pi_\tau}}$
            \STATE Update for all $i, t, x, u$
            \STATE $y_{t, \tau + 1}^i (x, u) = y_{t, \tau}^i (x, u) + \gamma \cdot Q_i^{\pi_\tau^i, \mu^{\pi_\tau}} (x,a)$
            \STATE $\pi_{t, \tau +1}^i (\cdot \vert x) = \Gamma \bc{y_{t, \tau + 1}^i (x, \cdot)}$
        \ENDFOR
    \end{algorithmic}
\end{algorithm}

Then, we can leverage this technical observation to derive the following result which will be crucial for proving Theorem~\ref{thm:unique_conv} from the main text.
\begin{lemma}\label{lem:derivative_H}
	The similarity measure $H$ defined in equation \eqref{def:sim_measure} satisfies
	\begin{align*}
		\frac{\mathrm d}{\mathrm d \tau} H \bc{y_\tau} = \Delta J \bc{\boldsymbol{\pi}_\tau, \boldsymbol{\pi}^*} + \tilde{d}\bc{\boldsymbol{\pi}_\tau, \boldsymbol{\pi}^*} 
	\end{align*}
where $\Delta J \bc{\boldsymbol{\pi}_\tau, \boldsymbol{\pi}^*} = \int_{\mathcal{I}} J_\alpha^{\boldsymbol{\mu}^*} \bc{\pi_\tau^\alpha} -J_\alpha^{\boldsymbol{\mu}^*} \bc{\pi^{\alpha,*}} \mathrm d \alpha$ is always non-positive with $\boldsymbol \mu^* = \Psi(\boldsymbol \pi^*)$ and $\tilde{d}\bc{\boldsymbol{\pi}, \boldsymbol{\pi'}}$ defined as in equation \eqref{def:mon_metric}.
\end{lemma}
Note that the proofs of Lemma~\ref{lem:derivative_con_conj} and  Lemma~\ref{lem:derivative_H} can be found in the subsequent sections.

\section{PROOF OF LEMMA~\ref{lem:omd_unique}}
Note that the following proof is based on an argument structure in \cite{perolat2021scaling}.
\begin{proof}
 Theorem~\ref{thm:existence} ensures that a NE exists. So we prove uniqueness by contradiction. Therefore, assume there are two different Nash Equilibria $\boldsymbol \pi$ and $\boldsymbol \pi'$. Their associated mean fields are denoted by $\boldsymbol \mu = \Psi(\boldsymbol \pi)$ and $\boldsymbol \mu' = \Psi(\boldsymbol \pi')$ as before. On the one hand, the strictly weak monotonicity of the mean field game implies 
\begin{align*}
	&\int_{\mathcal{I}} \left[ J^{\boldsymbol \mu}_{\alpha}\bc{\pi^\alpha} + J^{\boldsymbol \mu'}_{\alpha}\bc{{\pi'}^\alpha}  -  J^{\boldsymbol \mu}_{\alpha}\bc{{\pi'}^\alpha} -J^{\boldsymbol \mu'}_{\alpha}\bc{\pi^\alpha} \right] \mathrm d \alpha < 0.
\end{align*}
On the other hand, we know that both $\int_{\mathcal{I}} J^{\boldsymbol \mu}_{\alpha}\bc{\pi^\alpha} - J^{\boldsymbol \mu}_{\alpha}\bc{{\pi'}^\alpha} \mathrm d \alpha  \geq 0$ and $\int_{\mathcal{I}} J^{\boldsymbol \mu'}_{\alpha}\bc{{\pi'}^\alpha}  -J^{\boldsymbol \mu'}_{\alpha}\bc{\pi^\alpha} \mathrm d \alpha \geq 0$ hold by assuming that $\boldsymbol \pi$ and $\boldsymbol \pi'$ are Nash Equilibria. This in turn yields
\begin{align*}
	&\int_{\mathcal{I}} \left[ J^{\boldsymbol \mu}_{\alpha}\bc{\pi^\alpha} + J^{\boldsymbol \mu'}_{\alpha}\bc{{\pi'}^\alpha} - J^{\boldsymbol \mu}_{\alpha}\bc{{\pi'}^\alpha} -J^{\boldsymbol \mu'}_{\alpha}\bc{\pi^\alpha} \right] \mathrm d \alpha \geq 0.
\end{align*}
which is an obvious contradiction to the strictly weak monotonicity. This concludes the proof.
\end{proof}

\section{PROOF OF LEMMA~\ref{lem:derivative_con_conj}}
\begin{proof}
Recall that $h^*$ is defined as a function $h^*\colon \mathbb{R}^{\abs{\mathcal{U}}} \to \mathbb{R}$. For notational convenience, we define $h^* \bc{a_1, \ldots, a_{\abs{\mathcal{U}}}} = z \in \mathbb{R}$ where $a_1, \ldots, a_{\abs{\mathcal{U}}} \in \mathbb{R}$ and each $a_j$ represents the probability weight of one specific action $u_j \in \mathcal{U}$. Furthermore, we denote by $\frac{\mathrm d}{\mathrm d \tau}$ the total differential with respect to $\tau$. These definitions allow us to state
\begin{align*}
	&\frac{\mathrm d}{\mathrm d \tau}h^* \bc{y_{t, \tau}^\alpha \bc{x^\alpha, \cdot}} = \sum_{u_j \in \mathcal U} \frac{\partial h^*}{\partial a_j}  \bc{y_{t, \tau}^\alpha \bc{x^\alpha, u_j}} \frac{\mathrm d y_{t, \tau}^\alpha \bc{x^\alpha, u_j} }{\mathrm d \tau}.
\end{align*}
Keeping in mind the OMD algorithm and especially
\begin{align*}
	y_{t, \tau}^\alpha \bc{x^\alpha, u^\alpha} = \int_0^\tau Q_\alpha^{\boldsymbol{\mu}^{\boldsymbol{\pi}_s}, \pi^\alpha_s} \bc{t, x^\alpha, u^\alpha} \mathrm d s
\end{align*}
yields
\begin{align*}
	 \frac{\mathrm d y_{t, \tau}^\alpha \bc{x^\alpha, u_j} }{\mathrm d \tau} = Q_i^{\boldsymbol{\mu}^{\boldsymbol{\pi}_\tau}, \pi^\alpha_\tau} \bc{t, x^\alpha, u^\alpha}.
\end{align*}
Additionally, by definition we have
\begin{align*}
	\Gamma(y) = \nabla h^* (y) \quad \textrm{ and } \quad \pi_{t, \tau}^\alpha \bc{\cdot \big\vert x^\alpha} = \Gamma \bc{y_{t, \tau}^\alpha \bc{x^\alpha, \cdot}}
\end{align*}
which yields
\begin{align*}
	 \frac{\partial h^*}{\partial a_j} \bc{y_{t, \tau}^\alpha \bc{x^\alpha, u_j}} =  \Gamma \bc{y_{t, \tau}^\alpha \bc{x^\alpha, u_j}} =  \pi_{t, \tau}^\alpha \bc{u_j \big\vert x^\alpha}.
\end{align*}
In summary, the above observations imply
\begin{align*}
	\frac{\mathrm d}{\mathrm d \tau} h^* \bc{y_{t, \tau}^\alpha \bc{x^\alpha, \cdot}} 
	&= \sum_{u_j \in \mathcal U} \frac{\partial h^*}{\partial a_j} \frac{\mathrm d y_{t, \tau}^\alpha \bc{x^\alpha, u_j} }{\mathrm d \tau}\\
	&\qquad = \sum_{u_j \in \mathcal U}\pi_{t, \tau}^\alpha \bc{u_j \big\vert x^\alpha}  Q_i^{\boldsymbol{\mu}^{\boldsymbol{\pi}_\tau}, \pi^\alpha_\tau} \bc{t, x^\alpha, u^\alpha}
	= \scal{\pi_{t, \tau}^\alpha \bc{\cdot \big\vert x^\alpha}}{Q_i^{\boldsymbol{\mu}^{\boldsymbol{\pi}_\tau}, \pi^\alpha_\tau} \bc{t, x^\alpha, \cdot}}
\end{align*}
which yields the desired result.
\end{proof}

\section{PROOF OF LEMMA \ref{lem:derivative_H}}\label{sec:last_lemma_OMD}
This section follows \cite[Appendix D]{perolat2021scaling}.
\begin{proof}
	Assume that $\boldsymbol{\pi}^* \in \boldsymbol{\Pi}$ is a NE. Furthermore, for all $\tau > 0$ define the mean field $\boldsymbol \mu'_{ \tau} = \Psi \bc{\boldsymbol \pi_{ \tau}}$ and keep in mind that the transition kernel does not depend on the mean field by assumption. Then we have
	\begin{align*}
		&\frac{\mathrm d}{\mathrm d \tau} H \bc{y_\tau}\\
		&= \int_{\mathcal{I}} \sum_{t \in \mathcal T} \sum_{x^\alpha \in \mathcal X} \mu_t^{\alpha, \boldsymbol{\pi}^*} \bc{x^\alpha} \frac{\mathrm d}{\mathrm d \tau} \left[ h^* \bc{y_{t, \tau}^\alpha \bc{x^\alpha, \cdot}} - h^* \bc{y_t^{\alpha, *}\bc{x^\alpha, \cdot}} - \scal{\pi_t^{\alpha, *} \bc{\cdot \vert x^\alpha}}{y_{t, \tau}^\alpha \bc{x^\alpha, \cdot} - y_t^{\alpha, *}\bc{x^\alpha, \cdot}} \right] \mathrm d \alpha\\
		&= \int_{\mathcal{I}} \sum_{t \in \mathcal T} \sum_{x^\alpha \in \mathcal X} \mu_t^{\alpha, \boldsymbol{\pi}^*} \bc{x^\alpha} \cdot \left[ \frac{\mathrm d}{\mathrm d \tau}h^* \bc{y_{t, \tau}^\alpha \bc{x^\alpha, \cdot}} - \scal{\pi_t^{\alpha, *}\bc{\cdot \vert x^\alpha}}{\frac{\mathrm d}{\mathrm d \tau} y_{t, \tau}^\alpha \bc{x^\alpha, \cdot}} \right] \mathrm d \alpha\\
		&= \int_{\mathcal{I}} \sum_{\substack{t \in \mathcal T \\ x^\alpha \in \mathcal X}} \mu_t^{\alpha, \boldsymbol{\pi}^*} \bc{x^\alpha}  \left[\scal{\pi^{\alpha}_{t, \tau}\bc{\cdot \vert x^\alpha}}{Q_\alpha^{ \boldsymbol \mu'_{ \tau}, \pi_\tau^\alpha}\bc{t, x^\alpha, \cdot}} - \scal{\pi_t^{\alpha, *}\bc{\cdot \vert x^\alpha}}{Q_i^{\boldsymbol \mu'_{ \tau}, \pi_\tau^\alpha}\bc{t, x^\alpha, \cdot}} \right] \mathrm d \alpha\\
		&= \int_{\mathcal{I}} \sum_{t \in \mathcal T} \sum_{x^\alpha \in \mathcal X} \mu_t^{\alpha, \boldsymbol{\pi}^*} \bc{x^\alpha}  \brk{V_\alpha^{\boldsymbol \mu'_{ \tau}, \pi_\tau^\alpha}\bc{t, x^\alpha} - \scal{\pi_t^{\alpha, *}\bc{\cdot \vert x^\alpha}}{r \bc{x^\alpha, \cdot, \mu'_{t, \tau}}}} \mathrm d \alpha\\
		&\quad- \int_{\mathcal{I}} \sum_{t \in \mathcal T} \sum_{x^\alpha \in \mathcal X} \mu_t^{\alpha, \boldsymbol{\pi}^*} \bc{x^\alpha} \scal{\pi_t^{\alpha, *}\bc{\cdot \vert x^\alpha}}{\sum_{{x'}^\alpha \in \mathcal X} P \bc{{x'}^\alpha \big\vert x^\alpha, \cdot} V_\alpha^{\boldsymbol \mu'_{ \tau}, \pi_\tau^\alpha}\bc{t + 1, {x'}^\alpha}} \mathrm d \alpha\\
		&= \int_{\mathcal{I}} \sum_{t \in \mathcal T} \sum_{x^\alpha \in \mathcal X} \brk{\mu_t^{\alpha, \boldsymbol{\pi}^*} \bc{x^\alpha}  V_\alpha^{\boldsymbol \mu'_{ \tau}, \pi_\tau^\alpha}\bc{t, x^\alpha}} - \sum_{\substack{t \in \mathcal T \\ x^\alpha \in \mathcal X}} \mu_t^{\alpha, \boldsymbol{\pi}^*} \bc{x^\alpha} \scal{\pi_t^{\alpha, *}\bc{\cdot \vert x^\alpha}}{r \bc{x^\alpha, \cdot, \mu'_{t, \tau}}} \\
		&\quad \qquad \qquad \qquad \qquad \qquad \qquad \qquad \qquad \qquad \qquad \qquad -  \sum_{t \in \mathcal T} \sum_{{x'}^\alpha \in \mathcal X}V_\alpha^{\boldsymbol \mu'_{ \tau}, \pi_\tau^\alpha}\bc{t + 1, {x'}^\alpha}  \mu_{t + 1}^{\alpha, \boldsymbol{\pi}^*} \bc{{x'}^\alpha} \mathrm d \alpha\\
		&= \int_{\mathcal{I}} \left( \sum_{t \in \mathcal T} \sum_{x^\alpha \in \mathcal X} \brk{\mu_t^{\alpha, \boldsymbol{\pi}^*} \bc{x^\alpha}  V_\alpha^{\boldsymbol \mu'_{ \tau}, \pi_\tau^\alpha}\bc{t, x^\alpha}} -  \sum_{t \in \mathcal T} \sum_{x^\alpha \in \mathcal X} \brk{\mu_{t+1}^{\alpha, \boldsymbol{\pi}^*} \bc{x^\alpha}  V_i^{\boldsymbol \mu'_{ \tau}, \pi_\tau^\alpha}\bc{t + 1, x^\alpha}}\right.  \\
		& \quad \qquad \qquad \qquad \qquad \qquad \qquad \qquad \qquad \qquad \qquad \qquad \left. - \sum_{\substack{t \in \mathcal T \\ x^\alpha \in \mathcal X}} \mu_t^{\alpha, \boldsymbol{\pi}^*} \bc{x^\alpha} \scal{\pi_t^{\alpha, *}\bc{\cdot \vert x^\alpha}}{r \bc{x^\alpha, \cdot, \mu'_{t, \tau}}} \right) \mathrm d \alpha\\
		&= \int_{\mathcal{I}} \brk{J_\alpha^{\boldsymbol{\mu'}_\tau}\bc{\pi_\tau^\alpha} - J_\alpha^{\boldsymbol{\mu'}_\tau} \bc{\pi^{\alpha, *}}} \mathrm d \alpha\\
		&= \int_{\mathcal{I}} \left[J_\alpha^{\boldsymbol{\mu'}_\tau}\bc{\pi_\tau^\alpha} - J_\alpha^{\boldsymbol{\mu'}_\tau} \bc{\pi^{\alpha, *}} +J_i^{\boldsymbol{\mu}^*} \bc{\pi^{\alpha, *}} -J_\alpha^{\boldsymbol{\mu}^*} \bc{\pi_\tau^\alpha} -J_\alpha^{\boldsymbol{\mu}^*} \bc{\pi^{\alpha, *}} + J_\alpha^{\boldsymbol{\mu}^*} \bc{\pi_\tau^\alpha}\right] \mathrm d \alpha \\
		&= \Delta J \bc{\boldsymbol{\pi}_\tau, \boldsymbol{\pi}^*} + \tilde{d}\bc{\boldsymbol{\pi}_\tau, \boldsymbol{\pi}^*} 
	\end{align*}
where the third equality follows from Lemma~\ref{lem:derivative_con_conj}.
\end{proof}

\section{PROOF OF THEOREM~\ref{thm:unique_conv}}\label{sec:final_proof_OMD}

The following proof is based on an idea in \cite[Appendix G]{perolat2021scaling}.
\begin{proof}
	Define the function $f: \boldsymbol{\Pi} \to \mathbb{R}$ given by
	\begin{align}
		f (\boldsymbol{\pi}) &= \int_{\mathcal{I}} \sum_{t = 0}^T \sum_{x^\alpha \in \mathcal{X}} \mu_t^{\alpha, \boldsymbol{\pi}^*} \bc{x^\alpha}  D_{\mathrm{KL}} \bc{\pi_t^{\alpha, *} \bc{\cdot \big\vert x^\alpha }, \pi_t^{\alpha} \bc{\cdot \big\vert x^\alpha }} \mathrm d \alpha \nonumber
	\end{align}
where $\boldsymbol{\pi}^*$ is the NE of the smoothed MP-MFG which is unique by Lemma~\ref{lem:omd_unique}. We point out that this immediately implies that $f (\boldsymbol{\pi}) = 0$ if and only if $\boldsymbol{\pi} = \boldsymbol{\pi}^*$ (almost everywhere) and $f (\boldsymbol{\pi}) > 0$ otherwise by the basic properties of the Kullback-Leibler divergence.

The just defined function $f$ and the function $H$ defined in \eqref{def:sim_measure} are closely related. To see that, we start with some calculations
\begin{align}\label{eq:H_equivalence}
&h^* \bc{y_{t, \tau}^i \bc{x^i, \cdot}} - h^* \bc{y_t^{i, *}\bc{x^i, \cdot}}
-\scal{\pi_t^{i, *} \bc{\cdot \vert x^i}}{y_{t, \tau}^i \bc{x^i, \cdot} - y_t^{i, *}\bc{x^i, \cdot}} \nonumber\\
&\quad = \log \bc{\sum_{u \in \mathcal{U}} \exp(y_{t, \tau}^i \bc{x^i, u})} - \log \bc{\sum_{u \in \mathcal{U}} \exp(y_t^{i, *}\bc{x^i, u})}  -  \scal{\pi_t^{i, *} \bc{\cdot \vert x^i}}{y_{t, \tau}^i \bc{x^i, \cdot} - y_t^{i, *}\bc{x^i, \cdot}}\nonumber \\
&\quad = \log \bc{\sum_{u \in \mathcal{U}} \exp(y_{t, \tau}^i \bc{x^i, u})} - \scal{\pi_t^{i, *} \bc{\cdot \big\vert x^i}}{y_{t, \tau}^i (x^i, \cdot)} + \sum_{u \in \mathcal{U}} \pi_t^{i, *} \bc{u \vert x^i}  \log \bc{\frac{\log (y_t^{i, *}\bc{x^i, u})}{\sum_{u' \in \mathcal{U}} \exp(y_t^{i, *}\bc{x^i, u'})}} \nonumber\\
&\quad= \log \bc{\sum_{u' \in \mathcal{U}} \exp(y_{t, \tau}^i (x^i, u'))} - \scal{\pi_t^{i, *} \bc{\cdot \big\vert x^i}}{y_{t, \tau}^i (x^i, \cdot)} + \sum_{u \in \mathcal{U}} \pi_t^{i, *} \bc{u \big\vert x^i} \cdot  \log \bc{\pi_t^{i, *} \bc{u \big\vert x^i}} \nonumber\\
&\quad= \sum_{u \in \mathcal{U}} \pi_t^{i, *} \bc{u \big\vert x^i} \cdot \left[ \log \bc{\pi_t^{i, *} \bc{u \big\vert x^i}} + \log \bc{\frac{\sum_{u' \in \mathcal{U}} \exp(y_{t, \tau}^i (x^i, u'))}{\exp(y_{t, \tau}^i (x^i, u))}} \right] \nonumber\\
&\quad= \sum_{u \in \mathcal{U}} \pi_t^{i, *} \bc{u \big\vert x^i} \cdot \log \bc{\frac{\pi_t^{i, *} \bc{u \big\vert x^i}}{\Gamma \bc{y_{t, \tau}^i (x^i, u)}}} \nonumber\\
&\quad= D_{\mathrm{KL}} \bc{\pi_t^{i, *} \bc{\cdot \big\vert x^i }, \Gamma \bc{y_{t, \tau}^i (x^i, \cdot)}}. \nonumber
\end{align}
Therefore, we have $H(y) = f(\boldsymbol{\pi})$ for all pairs $y, \boldsymbol{\pi}$ for which $\Gamma(y) = \boldsymbol{\pi}$ holds. This, in turn implies together with Lemma~\ref{lem:derivative_con_conj} that for an arbitrary $\boldsymbol{\pi}_\tau \in \boldsymbol{\pi}$ we have
\begin{align}
	\frac{\mathrm d}{\mathrm d \tau} f(\boldsymbol{\pi}_\tau)
	= \Delta J \bc{\boldsymbol{\pi}_\tau, \boldsymbol{\pi}^*} + \tilde{d}\bc{\boldsymbol{\pi}_\tau, \boldsymbol{\pi}^*}
\end{align}
Now, we claim that 
\begin{align}\label{eq:der_f_is_non_positive}
	\frac{\mathrm d}{\mathrm d \tau} f(\boldsymbol{\pi}_\tau)
	= \Delta J \bc{\boldsymbol{\pi}_\tau, \boldsymbol{\pi}^*} + \tilde{d}\bc{\boldsymbol{\pi}_\tau, \boldsymbol{\pi}^*} \leq 0
\end{align}
is true for all $\boldsymbol{\pi}_\tau \in \boldsymbol{\pi}$ and that it is equal to zero if and only $\boldsymbol{\pi}_\tau$ is the NE, i.e. $\boldsymbol{\pi}_\tau = \boldsymbol{\pi}^*$.
To see this, consider an arbitrary but fixed $\boldsymbol{\pi}_\tau \in \boldsymbol{\pi}$. The term $\Delta J \bc{\boldsymbol{\pi}_\tau, \boldsymbol{\pi}^*} + \tilde{d}\bc{\boldsymbol{\pi}_\tau, \boldsymbol{\pi}^*}$ can either be negative or equal to zero. If $\Delta J \bc{\boldsymbol{\pi}_\tau, \boldsymbol{\pi}^*} + \tilde{d}\bc{\boldsymbol{\pi}_\tau, \boldsymbol{\pi}^*} = 0$ holds, this implies $ \tilde{d}\bc{\boldsymbol{\pi}_\tau, \boldsymbol{\pi}^*} = 0$ which in turn yields $\boldsymbol{\mu}^{\boldsymbol \pi_\tau} =\boldsymbol{\mu}^{\boldsymbol \pi^*}$. Besides that, $\Delta J \bc{\boldsymbol{\pi}_\tau, \boldsymbol{\pi}^*} + \tilde{d}\bc{\boldsymbol{\pi}_\tau, \boldsymbol{\pi}^*} = 0$ also means  that $\Delta J \bc{\boldsymbol{\pi}_\tau, \boldsymbol{\pi}^*} = 0$ has to hold. Reformulating the equation, we immediately obtain $\int_{\mathcal{I}} J_\alpha^{\boldsymbol{\mu}^*} \bc{\pi_\tau^\alpha} \mathrm d \alpha = \int_{\mathcal{I}} J_\alpha^{\boldsymbol{\mu}^*} \bc{\pi^{\alpha,*}} \mathrm d \alpha$. In combination with the previous observation $\boldsymbol{\mu}^{\boldsymbol \pi_\tau} =\boldsymbol{\mu}^{\boldsymbol \pi^*}$, this implies that $\boldsymbol{\pi}_\tau$ is a NE, i.e. $\boldsymbol{\pi}^*= \boldsymbol{\pi}_\tau$. To see this, recall that $\boldsymbol{\pi}^*$ is a NE by assumption and thereby a best response to $\boldsymbol{\mu}^{\boldsymbol \pi^*}$ (and $\boldsymbol{\mu}^{\boldsymbol \pi_\tau}$ by the first argument). Then,  $\int_{\mathcal{I}} J_\alpha^{\boldsymbol{\mu}^*} \bc{\pi_\tau^\alpha} \mathrm d \alpha = \int_{\mathcal{I}} J_\alpha^{\boldsymbol{\mu}^*} \bc{\pi^{\alpha,*}} \mathrm d \alpha$ means that both $\boldsymbol{\pi}^*$ and $\boldsymbol{\pi}_\tau$ yield the same expected reward given $\boldsymbol{\mu}^{\boldsymbol \pi_\tau}$. Since $\boldsymbol{\pi}^*$ is a best response to $\boldsymbol{\mu}^{\boldsymbol \pi_\tau}$, $\boldsymbol{\pi}_\tau$ also has to be a best response to $\boldsymbol{\mu}^{\boldsymbol \pi_\tau}$ and therefore a NE.

However, if $\Delta J \bc{\boldsymbol{\pi}_\tau, \boldsymbol{\pi}^*} + \tilde{d}\bc{\boldsymbol{\pi}_\tau, \boldsymbol{\pi}^*} < 0$ holds, $\boldsymbol{\pi}$ cannot be a NE. Assume, by contradiction, that $\boldsymbol{\pi}$ is a NE. Then, we have
\begin{align*}
	&\tilde{d}\bc{\boldsymbol{\pi}_\tau, \boldsymbol{\pi}^*} =\int_{\mathcal{I}} \left[ J^{\boldsymbol \mu_\tau}_{\alpha}\bc{\pi^\alpha_\tau}  - J^{\boldsymbol \mu_\tau}_{\alpha}\bc{\pi^{\alpha, *}} + J^{\boldsymbol \mu^*}_{\alpha}\bc{\pi^{\alpha, *}} -J^{\boldsymbol \mu^*}_{\alpha}\bc{\pi^i_\tau}\right]  \mathrm d \alpha \geq 0
\end{align*}
because $\boldsymbol{\pi}_\tau$ and $\boldsymbol{\pi}^*$ which implies $\int_{\mathcal{I}} J^{\boldsymbol \mu_\tau}_{\alpha}\bc{\pi^\alpha_\tau}  - J^{\boldsymbol \mu_\tau}_{\alpha}\bc{\pi^{\alpha, *}} \mathrm d \alpha \geq 0$ and $\int_{\mathcal{I}}  J^{\boldsymbol \mu^*}_{\alpha}\bc{\pi^{\alpha, *}} -J^{\boldsymbol \mu^*}_{\alpha}\bc{\pi^\alpha_\tau} \mathrm d \alpha \geq 0$, respectively. This observation combined with inequality \eqref{def:mon_metric} yields $\tilde{d}\bc{\boldsymbol{\pi}_\tau, \boldsymbol{\pi}^*} = 0$ and thereby $\boldsymbol{\mu}^{\boldsymbol \pi_\tau} =\boldsymbol{\mu}^{\boldsymbol \pi^*}$. Thus, $\boldsymbol{\pi}_\tau$ and $\boldsymbol{\pi}^*$ are both best responses to $\boldsymbol{\mu}^{\boldsymbol \pi^*}$ which means that $\Delta J \bc{\boldsymbol{\pi}_\tau, \boldsymbol{\pi}^*} =0$. This is a contradiction to the assumption $\Delta J \bc{\boldsymbol{\pi}_\tau, \boldsymbol{\pi}^*} + \tilde{d}\bc{\boldsymbol{\pi}_\tau, \boldsymbol{\pi}^*} < 0$. Hence, $\Delta J \bc{\boldsymbol{\pi}_\tau, \boldsymbol{\pi}^*} + \tilde{d}\bc{\boldsymbol{\pi}_\tau, \boldsymbol{\pi}^*} < 0$ implies that $\boldsymbol{\pi}_\tau$ is not a NE.

With the above arguments in place, we can apply LaSalle's theorem \cite[Theorem 4.4]{khalil2002nonlinear} to solve the differential equation \eqref{eq: CTOMD_diff_eq}. Choosing $f ( \boldsymbol{y})$ as the function $V$ in \cite[Theorem 4.1]{khalil2002nonlinear} with  $\boldsymbol{\Pi}$ as the compact set, we point out that the unique NE $\boldsymbol{\pi}^*$ is the only element in $\boldsymbol{\Pi}$ with $\frac{\mathrm d}{\mathrm d \tau} f(\boldsymbol{\pi}^*) = 0$ as we have established in equation \eqref{eq:der_f_is_non_positive}. Therefore, the OMD algorithm converges to $\boldsymbol{\pi}^*$ for every starting point $\boldsymbol{\pi}_0 \in \boldsymbol{\Pi}$.
\end{proof}

\section{CUTOFF POWER LAW GRAPHON}\label{app:cut_off}

In this paragraph we consider graphs $G_n$ with cutoff power law degree distributions. The benefits of adding a cutoff are twofold. On the one hand, a technical advantage of the modified degree distribution is that it is both integrable and Lipschitz continuous. On the other hand, the cutoff version of the power-law turns out to be a more realistic modelling option in numerous real-world applications. In practice, the expression 'power law distribution' frequently refers only to the tail of the distribution. This accounts for the fact that pure power laws often do not describe the empirical observations accurately, see for example \cite{newman2005power}. The interested reader is also referred to \cite{newman2018networks} and the references therein for an overview on the topic.

The cutoff power-law distribution can be constructed by connecting two vertices $i, j \in \cbc{1, \ldots, n}$ with probability $\min \bc{1, n^\beta \bc{\max\cbc{i, c \cdot n} \cdot \max \cbc{j, c \cdot n}}^{- \alpha}}$ independently of all other edges where $0 < c < 1$, $0 < \alpha < 1$, and $0 \leq \beta < 2 \alpha$. Furthermore, the assumption $\beta > 2 \alpha -1$ implies that the expected number of edges is super-linear. We note that $n^\beta \bc{ c \cdot n \cdot c \cdot n}^{- \alpha} = c^{ - 2 \alpha} n^{\beta - 2 \alpha} = o(1)$ is implied by the assumption $\beta < 2 \alpha$. This in turn allows us to drop the $\min$ term in the edge probability since for $n \rightarrow \infty$ we have $n^\beta \bc{\max\cbc{i, c \cdot n} \cdot \max \cbc{j, c \cdot n}}^{- \alpha} = o(1)$. The following calculation provides the expected edge density which is a key element for the subsequent arguments.
\begin{align*}
	&\Erw \brk{\norm{G_n}_1} = 
	n^{- 2} \sum_{i, j \in V\bc{G_n}} \Erw \brk{ \mathbf{1}_{(i j) \in E \bc{G_n}}} \\
	&= \sum_{i, j \in V\bc{G_n}}\frac{\min \bc{1, n^\beta \bc{\max\cbc{i, c n} \cdot \max \cbc{j, c n}}^{- \alpha}}}{n^2}\\
	&= n^{- 2} \sum_{i, j \in V\bc{G_n}} n^\beta \bc{\max\cbc{i, c \cdot n} \cdot \max \cbc{j, c \cdot n}}^{- \alpha} + o \bc{n^{-2}} \\
	&= n^{\beta - 2} \bc{\sum_{i, j \in V\bc{G_n} } \bc{c \cdot n}^{- 2 \alpha} \mathbf{1}_{\cbc{i, j \leq c \cdot n}}} +  o \bc{n^{-2}}  + n^{\beta - 2} \bc {\sum_{i, j \in V\bc{G_n}}  (i j)^{- \alpha} \mathbf{1}_{\cbc{i, j > c \cdot n}}} \\
	&\qquad \qquad \qquad \qquad \qquad \qquad \qquad \qquad \qquad + 2 n^{\beta - 2} \bc {\sum_{i, j \in V\bc{G_n}}  (i \cdot c \cdot n)^{- \alpha} \mathbf{1}_{\cbc{i > c \cdot n \geq j}}} \\
	&= c ^{- 2 \alpha} n^{\beta - 2 - 2 \alpha} (c  n)^2 +  o \bc{n^{-2}} + n^{\beta - 2} \bc {\sum_{i, j \in V\bc{G_n}}  (i j)^{- \alpha} \mathbf{1}_{\cbc{i, j > c  n}}} + 2 c^{- \alpha} \cdot n^{\beta - 2 - \alpha} c  n \bc {\sum_{i \in V\bc{G_n}}  i ^{- \alpha} \mathbf{1}_{\cbc{i > c \cdot n }}} \\
	&\sim c ^{2 - 2 \alpha} n^{\beta - 2 \alpha} + n^{\beta - 2} \bc { \int_{c \cdot n}^n i^{-\alpha} \, \mathrm d i} \bc { \int_{c \cdot n}^n j^{-\alpha} \, \mathrm d j} + 2 c^{1 - \alpha}  n^{\beta - 1 - \alpha} \int_{c \cdot n}^n i^{-\alpha} \, \mathrm d i \\
	&= c ^{2 - 2 \alpha} n^{\beta - 2 \alpha} + n^{\beta - 2}(1 - \alpha)^{-2} \bc{n^{1 - \alpha} - (c n)^{1 - \alpha}}^2\\
	&\quad + 2 c^{1 - \alpha}  n^{\beta - 1 - \alpha} (1 - \alpha)^{-1} \bc{n^{1 - \alpha} - (c n)^{1 - \alpha}} \\
	&= c ^{2 - 2 \alpha} n^{\beta - 2 \alpha} + n^{\beta - 2 \alpha}(1 - \alpha)^{-2}\bc{1 - c^{1 - \alpha}}^2  + 2 c^{1 - \alpha}  n^{\beta - 2 \alpha} (1 - \alpha)^{-1} \bc{1 - c^{1 - \alpha}}\\
	&=n^{\beta - 2 \alpha}   c ^{2 - 2 \alpha}  +  n^{\beta - 2 \alpha}\frac{1 - c^{1 - \alpha}}{(1 - \alpha)^2} \brk{1 - c^{1 - \alpha} + 2 (1- \alpha) c^{1 - \alpha}} \\
	&=(1 - \alpha)^{-2} n^{\beta - 2 \alpha}  \left( c ^{2 - 2 \alpha} (1 - \alpha)^2 + 1 - c^{2 - 2\alpha}  - 2  \alpha c^{1 - \alpha  }\bc{1 - c^{1 - \alpha}}\right) \\
	&=(1 - \alpha)^{-2} n^{\beta - 2 \alpha}  \bc{ \alpha^2 c ^{2 - 2 \alpha} + 1- 2  \alpha c^{1 - \alpha }}\\
	&=(1 - \alpha)^{-2} n^{\beta - 2 \alpha}  \bc{ 1 - \alpha c^{1 - \alpha}}^2
\end{align*}

Knowing the expected edge density, we are now able to determine the limiting graphon of interest. The required mathematical framework can be found in \cite{borgs2018p} which also provides the idea for our proof. Therefore, we just give the key definitions and refer to \cite{borgs2018p} for more details. In general, consider a weighted graph $G$ and a partition $\mathcal{P}$ of $G$ into $q$ parts, i.e. $\mathcal{P} \coloneqq \cbc{ V_1, \ldots , V_q}$ where some sets of the partition can be empty. Alternatively, this partition can be expressed in terms of a map $\phi \colon V(G) \rightarrow [q]$ where $\phi (x) = i$ if and only if $x \in V_i$. This map $\phi$ gives rise to the definition of the quotient $G/ \phi$ which consists of a pair $\bc{\alpha \bc {G/ \phi}, \beta \bc{G/ \phi}}$ where $\alpha \bc {G/ \phi} \in \mathbb{R}^q$ and $\beta \bc{G/ \phi} \in \mathbb{R}^{q \times q}$. The entries of $\alpha \bc {G/ \phi}$ and $\beta \bc{G/ \phi}$ are defined as
\begin{align*}
	\alpha_i \bc {G/ \phi} \coloneqq \frac{\alpha_{V_i} (G)}{\alpha_G}
\end{align*}
and
\begin{align*}
	\beta_{i j} \bc {G/ \phi} \coloneqq \frac{1}{\norm{G}_1} \sum_{(u,v) \in V_i \times V_j} \frac{\alpha_u (G)}{\alpha_G}  \frac{\alpha_v (G)}{\alpha_G}  \beta_{u v} (G)
\end{align*}
for all $i, j \in \cbc{1, \ldots, q}$. To obtain analogous concepts for graphons, we define a fractional partition of $[0,1]$ into $q$ classes as a $q$-tuple of measurable functions $\rho_1, \ldots, \rho_q\colon [0, 1] \rightarrow [0,1]$ with $\rho_1 (x) + \ldots + \rho_q (x) =1$ for all $x \in [0,1]$. Then, the pair $\bc{\alpha_i (W / \rho), \beta (W / \rho)}$ with entries
\begin{align*}
	\alpha_i (W / \rho) &\coloneqq \alpha_i (\rho) \coloneqq \int_0^1 \rho_i (x) \, \mathrm d x \\
	\beta_{i j} (W / \rho) &\coloneqq  \int_{\brk{0, 1}^2} \rho_i (x) \rho_j (y) W(x, y) \, \mathrm d x \, \mathrm d y
\end{align*}
defines the quotient $W / \rho$. With the above definitions in place, we turn to the proof that the graph sequence $\bc{G_n}_{n \in \mathbb{N}}$ converges to the $L^p$ graphon $W(x,y) = (1 - \alpha)^2 \bc{ 1 - \alpha c^{1 - \alpha}}^{-2} \bc{\max\cbc{x, c} \cdot \max \cbc{y, c}}^{- \alpha}$.
For each $V_i$ we define
\begin{align*}
	A_i \coloneqq \bigcup_{x \in V_i} (x -1, x] \qquad \textrm{and} \qquad B_i \coloneqq \bigcup_{x \in V_i} \Bigg(\frac{x -1}{n}, \frac{x}{n}\Bigg]
\end{align*}
which implies both $A_i \subset [0, n]$ and $B_i \subset [0, 1]$.
Turn to the expectation
\begin{align*}
	\Erw \brk{\beta_{i j} \bc {G_n / \phi}} 
	&\sim \frac{1}{n^2 \Erw \brk{\norm{G}_1}} \sum_{(u,v) \in V_i \times V_j}  \min \bc{1, n^\beta \bc{\max\cbc{u, c \cdot n} \cdot \max \cbc{v, c \cdot n}}^{- \alpha}} \\
	&\sim \frac{(1 - \alpha)^2}{ n^{\beta - 2 \alpha}  \bc{ 1 - \alpha c^{1 - \alpha}}^2} \sum_{(u,v) \in V_i \times V_j}   n^{\beta - 2} \bc{\max\cbc{u, c \cdot n} \cdot \max \cbc{v, c \cdot n}}^{- \alpha}\\
	&= (1 - \alpha)^{2} n^{ 2 \alpha -2}  \bc{ 1 - \alpha c^{1 - \alpha}}^{- 2}  \sum_{(u,v) \in V_i \times V_j} \bc{\max\cbc{u, c \cdot n} \cdot \max \cbc{v, c \cdot n}}^{- \alpha}\\
	&\sim \bc{\frac{1 - \alpha}{1 - \alpha c^{1 - \alpha}}}^2 n^{ 2 \alpha -2} \cdot \int_{\brk{0, n}^2} \mathbf{1}_{\cbc{x \in A_i}} \mathbf{1}_{\cbc{y \in A_j}} \\ &\qquad \qquad \cdot \bc{\max\cbc{x, c \cdot n} \cdot \max \cbc{y, c \cdot n}}^{- \alpha} \, \mathrm d x \, \mathrm d y\\
	&= \bc{\frac{1 - \alpha}{1 - \alpha c^{1 - \alpha}}}^2 n^{ -2} \cdot \int_{\brk{0, n}^2} \mathbf{1}_{\cbc{x \in A_i}} \mathbf{1}_{\cbc{y \in A_j}}  \bc{\max\cbc{x/n, c} \cdot \max \cbc{y/n, c}}^{- \alpha} \, \mathrm d x \, \mathrm d y\\
	&= \int_{\brk{0, 1}^2} \mathbf{1}_{\cbc{x \in  B_i}} \mathbf{1}_{\cbc{y \in B_j}} (1 - \alpha)^2 \bc{ 1 - \alpha c^{1 - \alpha}}^{-2} \bc{\max\cbc{x, c} \cdot \max \cbc{y, c}}^{- \alpha} \, \mathrm d x \, \mathrm d y \\
	&= \int_{\brk{0, 1}^2} \rho_i (x) \rho_j (y) W(x, y) \, \mathrm d x \, \mathrm d y
\end{align*}
where
\begin{align*}
	 W(x, y) = \bc{\frac{1 - \alpha}{1 - \alpha c^{1 - \alpha}}}^2 \bc{\max\cbc{x, c} \cdot \max \cbc{y, c}}^{- \alpha}.
\end{align*}
Eventually, arguments as in \cite[Section 3.3.3]{borgs2018p} yield the desired convergence result.

\section{PROBLEM DETAILS}\label{app:prob_details}
In this section, we give more comprehensive descriptions of the problems considered in our work.

\subsection{Cyber Security} \label{app:cyber}
Let us start with fundamental components of the cyber security model. Computers are either infected or susceptible to infection. Furthermore, they can be defended and are otherwise called unprotected. This leads to a total number of four states, namely $DI$ (defended infected), $DS$ (defended and susceptible to infection), $UI$ (unprotected infected), and $US$ (unprotected and susceptible to infection). Formally, we define the state space as $\mathcal{X} \coloneqq \cbc{DI, DS, UI, US}$.  Each owner of a computer in the network can choose between two actions at each time step, i.e. $\mathcal{U} \coloneqq \cbc{0, 1}$. Here, action $0$ means that the computer owner is satisfied with its current defense status (defended or unprotected) and does not try to adjust it. Conversely, action $1$ captures the fact that the owner is trying to change the current defense status to the opposing one. It is important to note that a geometrically distributed (with parameter $0 < \lambda\leq 1$) waiting time passes before the desired adjustment takes place. This brings us to specifying the network structure of interest. In contrast to \cite{carmona2018probabilistic} we assume that the connections in the computer network are characterized by a not-so-dense graph structure. More specifically, we choose sequence of graphs that converge to a power-law $L^p$-graphon $W$ (see the previous chapters for details) when the number of computers in the network approaches infinity.  This gives rise to the usual empirical neighborhood mean field $\mathbb{G}_t^\alpha$  of agent $\alpha$ at time $t$.
\footnote{\cite{carmona2018probabilistic} define the model for a continuous time interval in contrast to our discrete time approach. Consequently, we make some adjustments such as using geometrically distributed waiting times instead of exponentially distributed ones.}

Already infected computers have a recovery rate which depends on their protection level, meaning there is one recovery probability $q_{\textrm{rec}}^D$ for defended computers and one for unprotected ones, i.e. $q_{\textrm{rec}}^U$. Computers can become infected either directly by the hacker or by other infected computers. Formally, we define $v_H$ to be the intensity of the hacker's attacks. Then, the probability for direct infection of a defended computer is $v_H z_{\textrm{inf}}^D$ while it is $v_H z_{\textrm{inf}}^U$ for an unprotected one. For the second way of infection through other infected computers, we assume that computers can only infect each other if they share a direct connection in the network. Put differently, the nodes depicting the respective computers in the graph have to be connected by an edge. This concept provides the foundation for the transition probabilities we have not specified yet. If $\alpha$ is an unprotected susceptible computer, it's probability to be infected by an infected and defended computer at time $t$ is $\beta_{D U} \mathbb{G}_t^\alpha \cbc{DI}$. Similarly, the unprotected susceptible computer $\alpha$ has a probability of $\beta_{U U} \mathbb{G}_t^\alpha \cbc{UI}$ to be infected by an unprotected and infected computer.

If we assume $\alpha$ to be a defended susceptible computer, the transition probabilities have to be adapted accordingly. The probability for $\alpha$ to be infected by a defended and infected computer at time $t$ is then given by $\beta_{D D} \mathbb{G}_t^\alpha \cbc{DI}$, and it's probability to be infected by an unprotected infected computer is defined as $\beta_{U D} \mathbb{G}_t^\alpha \cbc{UI}$. Summing up the probabilities of infection by the hacker and infection by another computer yields the overall probability of being infected, i.e.
\begin{align*}
    q_\textrm{inf}^D &\coloneqq v_H z_{\textrm{inf}}^D + \beta_{D D} \mathbb{\widetilde G}_t^\alpha \cbc{DI} + \beta_{U D} \mathbb{\widetilde G}_t^\alpha \cbc{UI} - v_H z_{\textrm{inf}}^D  \beta_{D D} \mathbb{\widetilde G}_t^\alpha \cbc{DI}
	- v_H z_{\textrm{inf}}^D   \beta_{U D} \mathbb{\widetilde G}_t^\alpha \cbc{UI} \\
	&\qquad \qquad - \beta_{D D} \mathbb{\widetilde G}_t^\alpha \cbc{DI}  \beta_{U D} \mathbb{\widetilde G}_t^\alpha \cbc{UI}  +v_H z_{\textrm{inf}}^D  \beta_{D D} \mathbb{\widetilde G}_t^\alpha \cbc{DI}  \beta_{U D} \mathbb{\widetilde G}_t^\alpha \cbc{UI}
\end{align*}
and similarly for undefended computers
\begin{align*}
    q_\textrm{inf}^U &\coloneqq v_H z_{\textrm{inf}}^U + \beta_{D U} \mathbb{\widetilde G}_t^\alpha \cbc{DI} + \beta_{U U} \mathbb{\widetilde G}_t^\alpha \cbc{UI}  -v_H z_{\textrm{inf}}^U  \beta_{D U} \mathbb{\widetilde G}_t^\alpha \cbc{DI}
	-v_H z_{\textrm{inf}}^U   \beta_{U U} \mathbb{\widetilde G}_t^\alpha \cbc{UI} \\
	&\qquad \qquad -  \beta_{D U} \mathbb{\widetilde G}_t^\alpha \cbc{DI}  \beta_{U U} \mathbb{\widetilde G}_t^\alpha \cbc{UI}  + v_H z_{\textrm{inf}}^U  \beta_{D U} \mathbb{\widetilde G}_t^\alpha \cbc{DI}  \beta_{U U} \mathbb{\widetilde G}_t^\alpha \cbc{UI}
\end{align*}
where $\mathbb{\widetilde G}_t^\alpha \coloneqq \min(1, \mathbb{G}_t^\alpha)$, and the superscript ($D$ or $U$) indicates whether the respective computer is defended or unprotected at the current time point. For convenience of notation, we do not explicitly indicate the dependence of $q_\textrm{inf}^D$ and $q_\textrm{inf}^U$ on $\alpha$ and $t$. Thus, we obtain the following state-transition matrix which specifies all transition probabilities
\begin{align}\label{cyber_transition_matrix}
&M \bc{\alpha, t, u}= \qquad \bordermatrix{~ & DI & DS & UI & US \cr
	DI & \Bar{\lambda} \Bar{q}_{\textrm{rec}}^D &  \Bar{\lambda} q_{\textrm{rec}}^D& u \lambda \Bar{q}_{\textrm{rec}}^D & u \lambda q_{\textrm{rec}}^D \cr
	DS & \Bar{\lambda} q_\textrm{inf}^D  &\Bar{\lambda} \Bar{q}_\textrm{inf}^D & u \lambda q_\textrm{inf}^D  & u \lambda  \Bar{q}_\textrm{inf}^D \cr
	UI & u \lambda \Bar{q}_{\textrm{rec}}^U & u \lambda q_{\textrm{rec}}^U & \Bar{\lambda} \Bar{q}_{\textrm{rec}}^U & \Bar{\lambda} q_{\textrm{rec}}^U\cr
	US & u \lambda  q_\textrm{inf}^U & u \lambda \Bar{q}_{\textrm{inf}}^U & \Bar{\lambda} q_\textrm{inf}^U & \Bar{\lambda} \Bar{q}_{\textrm{inf}}^U}
\end{align}
where we have introduced $\Bar{\lambda} \coloneqq 1 - u \lambda$ and $\Bar{q} \coloneqq 1 - q$ for all $q$ with the different respective subscripts and superscripts.
By choosing their controls $(u)_{t=0, \ldots, T - 1} \in \cbc{0,1}^{T}$, computer owners try to maximize their rewards given by the reward function $r \bc{X_{t}^i, U_{t}^i, \mathbb G^i_t} = - k_D \mathbf{1}_{\cbc{X_{t}^i \in D}} - k_I\mathbf{1}_{\cbc{X_{t}^i \in I}}$ where $D \coloneqq \cbc{DI, DS}$, $I \coloneqq \cbc{DI, UI}$, and  $k_D \geq 0$ ($k_I \geq 0$) is the cost associated with being defended (infected)  for one computer at one time step.

In our experiments, we use the parameters $T=50$, $\mu_0 = [0.25, 0.25, 0.25, 0.25]$, $q_{\textrm{rec}}^D = 0.3$,  $q_{\textrm{rec}}^U = 0.2$, $\lambda = 0.3$, $v_H = 0.1$, $z_{\textrm{inf}}^D = 0.05$, $z_{\textrm{inf}}^U = 0.1$, $\beta_{D D} = 0.1$, $\beta_{U D} = 0.2$, $\beta_{D U} = 0.7$, $\beta_{U U} = 0.8$, $k_D = 0.7$, $k_I = 2$.

\subsection{Heterogeneous Cyber Security} \label{app:heterogeneous}
We assume existence of two classes of computers, such as the ones in private ownership ($\mathrm{Pri}$) and the ones bought by corporations for business purposes. Let $\gamma_{\mathrm{Pri}}$ and $\gamma_{\mathrm{Cor}}$ be the fraction of private and corporate computers in the network, respectively, with $\gamma_{\mathrm{Pri}} + \gamma_{\mathrm{Cor}} = 1$. This distinction between classes can bring the model closer to reality for several reasons. First of all, the accessible level of protection as well as the associated costs may vary drastically between private users and corporations. While private computer owners usually have access to average protection measures at affordable price levels, corporations can decide to invest into large-scale protection solutions, raising both the level of protection and the costs. Second, the costs in case of an infection may differ between private and commercial users. For the private user, costs may consist of inconveniences, i.e. the inability to use one's computer properly, or financial losses such as hacked bank accounts. In contrast to that, infected corporate computers may lead to immense financial and economic damages while personal inconveniences are secondary in this case. Third, the infection probabilities might vary between commercial and private computers. This can either be caused by changes in the owner behavior due to the different usage context, i.e. doing work versus casual activities (listening to music, social media, etc.). Or it might be the case that the hacker chooses one class of computer owners as her primary target, e.g. because corporate computers promise higher financial gains.

The implementation of multiple classes into our model is a rather straightforward extension of the basic setup. We implement the two different agent types as part of the state, see also for example \cite{cui2021learning} or \cite{mondal2021approximation}. Then, the new state space is given by $\mathcal{X} \coloneqq \mathcal{X}_{P} \cup \mathcal{X}_{C}$ with $\mathcal{X}_{P} \coloneqq \cbc{\mathrm{Pri}DI, \mathrm{Pri}DS, \mathrm{Pri}UI, \mathrm{Pri}US}$ and  $\mathcal{X}_{C} \coloneqq \cbc{\mathrm{Cor}DI, \mathrm{Cor}DS,\mathrm{Cor} UI, \mathrm{Cor}US}$ while the action space $\mathcal{U} \coloneqq \cbc{0, 1}$ remains the same.
As the base version of the model, the multi-class setup is also built on the assumption that the computer network structure follows a graph with power-law degree distribution. Coming to the transition probabilities, the multi-class approach requires two versions of the transition matrix in $\eqref{cyber_transition_matrix}$. 

For various possible reasons, which we have discussed above, the infection and recovery probabilities for commercially used computers do not have to be the same as for privately owned ones. Formalizing this observation, we introduce the recovery probabilities for defended private computers $q_{\textrm{rec}}^{\textrm{Pri} D}$ and unprotected private computers $q_{\textrm{rec}}^{\textrm{Pri} U}$. These and the following definitions are easily extended to the corporate case. However, we will not write them down explicitly in order to keep this paragraph clear and brief. Similarly to the basic version of the model, a privately owned defended computer is directly infected by the hacker with probability $v_H z_{\textrm{inf}}^{\textrm{Pri} D}$ and a private unprotected one is directly infected with probability $v_H z_{\textrm{inf}}^{\textrm{Pri} U}$ where $v_H$ denotes the effort of the hacker as before. Besides direct infection, computers can also be infected by already infected, neighboring computers. To account for the potential differences of the two classes of private and corporate computers, the probability for an defended private computer $\alpha$ to be infected by an unprotected neighbor is denoted as $\beta_{U \textrm{Pri}D} \mathbb{G}_t^\alpha \cbc{UI}$ where $\mathbb{G}_t^\alpha$ denotes the empirical neighborhood mean field of agent $\alpha$ at time $t$ as before. Accordingly, we define the other network infection probabilities by  $\beta_{D \textrm{Pri}D} \mathbb{G}_t^\alpha \cbc{UI}$, $\beta_{D \textrm{Pri}U} \mathbb{G}_t^\alpha \cbc{UI}$, and $\beta_{U \textrm{Pri} U} \mathbb{G}_t^\alpha \cbc{UI}$.
Combining both the direct infection probabilities and those through the network, we obtain the overall infection probabilities
\begin{align*}
	q_\textrm{inf}^{\textrm{Pri} D} &\coloneqq 
	v_H z_{\textrm{inf}}^{\textrm{Pri} D} + \beta_{D \textrm{Pri} D} \mathbb{\widetilde G}_t^\alpha \cbc{DI}  + \beta_{U \textrm{Pri} D} \mathbb{\widetilde G}_t^\alpha \cbc{UI}  
	- v_H z_{\textrm{inf}}^{\textrm{Pri} D}  \beta_{D \textrm{Pri}D} \mathbb{\widetilde G}_t^\alpha \cbc{DI} - v_H z_{\textrm{inf}}^{\textrm{Pri} D}   \beta_{U\textrm{Pri} D} \mathbb{\widetilde G}_t^\alpha \cbc{UI} \\
	&\qquad \qquad - \beta_{D\textrm{Pri} D} \mathbb{\widetilde G}_t^\alpha \cbc{DI}  \beta_{U \textrm{Pri}D} \mathbb{\widetilde G}_t^\alpha \cbc{UI} +v_H z_{\textrm{inf}}^{\textrm{Pri} D} \beta_{D\textrm{Pri} D} \mathbb{\widetilde G}_t^\alpha \cbc{DI}  \beta_{U \textrm{Pri}D} \mathbb{\widetilde G}_t^\alpha \cbc{UI}
\end{align*}
and
\begin{align*}
	q_\textrm{inf}^{\textrm{Pri} U} &\coloneqq v_H z_{\textrm{inf}}^{\textrm{Pri} U} + \beta_{D \textrm{Pri} U} \mathbb{\widetilde G}_t^\alpha \cbc{DI}  + \beta_{U \textrm{Pri} U} \mathbb{\widetilde G}_t^\alpha \cbc{UI}
	-v_H z_{\textrm{inf}}^{\textrm{Pri} U}  \beta_{D\textrm{Pri} U} \mathbb{\widetilde G}_t^\alpha \cbc{DI} -v_H z_{\textrm{inf}}^{\textrm{Pri} U}   \beta_{U \textrm{Pri}U} \mathbb{\widetilde G}_t^\alpha \cbc{UI} \\
	&\qquad \qquad - \beta_{D \textrm{Pri} U} \mathbb{\widetilde G}_t^\alpha \cbc{DI}  \beta_{U \textrm{Pri}U} \mathbb{\widetilde G}_t^\alpha \cbc{UI} +  v_H z_{\textrm{inf}}^{\textrm{Pri} U}  \beta_{D \textrm{Pri} U} \mathbb{\widetilde G}_t^\alpha \cbc{DI}  \beta_{U \textrm{Pri} U} \mathbb{\widetilde G}_t^\alpha \cbc{UI}
\end{align*}
where $\mathbb{\widetilde G}_t^\alpha \coloneqq \min(1, \mathbb{G}_t^\alpha)$. With the above definitions, we provide the state-transition matrix $M_{\textrm{Pri}} \bc{\alpha, t, u}$ for the '$\textrm{Pri}$' class and similarly $M_{\textrm{Cor}} \bc{\alpha, t, u}$ for the '$\textrm{Cor}$' class analogously to the homogeneous problem. Eventually, the state-transition matrix $M$ for the whole state space $\mathcal{X}$ is given by the definition
\begin{align*}
	M \bc{\alpha, t, u} \coloneqq
	\begin{pmatrix}
		M_{\textrm{Pri}} \bc{\alpha, t, u} & \mathbf{0}\\
		\mathbf{0} & M_{\textrm{Cor}} \bc{\alpha, t, u}
	\end{pmatrix}
\end{align*}
where each $\mathbf{0}$ is a $4 \times 4$ matrix with all entries equal to zero. Since there is no possibility of changing ownership in our setup, the transition probabilities we have not specified yet all turn out to be zero. Computer owners try to maximize their rewards $r_{\textrm{Pri}}  \bc{X_{t}^i, U_{t}^i, \mathbb G^i_t} = - k_{\textrm{Pri} D} \mathbf{1}_{\cbc{X_{t}^i \in D}} - k_{\textrm{Pri} I} \mathbf{1}_{\cbc{X_{t}^i \in I}}$ or $r_{\textrm{Cor}} \bc{X_{t}^i, U_{t}^i, \mathbb G^i_t} = - k_{\textrm{Cor} D} \mathbf{1}_{\cbc{X_{t}^i \in D}} - k_{\textrm{Cor} I} \mathbf{1}_{\cbc{X_{t}^i \in I}}$, depending on the class of ownership. Here, $k_{\textrm{Pri} D}$ denotes the cost to defend a privately owned computer and $k_{\textrm{Pri} I}$ the cost associated with a privately owned, infected computer. Note that both costs are measured in terms of one time period and that $k_{\textrm{Cor} D}$ and $k_{\textrm{Cor} I}$ are defined analogously.

In our experiments, we use slightly adjusted parameters $T=50$, $\mu_0 = [0.125, 0.125, \ldots, 0.125]$, $q_{\textrm{rec}}^{\textrm{Cor} D} = 0.4$,  $q_{\textrm{rec}}^{\textrm{Cor} U} = 0.3$, $\lambda = 0.3$, $v_H = 0.1$, $z_{\textrm{inf}}^{\textrm{Cor} D} = 0.05$, $z_{\textrm{inf}}^{\textrm{Cor} U} = 0.1$, $\beta_{D \textrm{Cor}D} = 0.1$, $\beta_{U \textrm{Cor}D} = 0.2$, $\beta_{D \textrm{Cor}U} = 0.7$, $\beta_{U \textrm{Cor}U} = 0.8$, $k_{\textrm{Cor} D} = 0.7$, $k_{\textrm{Cor} I} = 2$, and parameters for the $\mathrm{Pri}$ agents $q_{\textrm{rec}}^{\textrm{Pri} D} = 0.4$,  $q_{\textrm{rec}}^{\textrm{Pri} U} = 0.3$, $\lambda = 0.3$, $v_H = 0.1$, $z_{\textrm{inf}}^{\textrm{Pri} D} = 0.05$, $z_{\textrm{inf}}^{\textrm{Pri} U} = 0.1$, $\beta_{D \textrm{Pri}D} = 0.2$, $\beta_{U \textrm{Pri}D} = 0.3$, $\beta_{D \textrm{Pri}U} = 0.9$, $\beta_{U \textrm{Pri}U} = 1.0$, $k_{\textrm{Pri} D} = 0.6$, $k_{\textrm{Pri} I} = 2$.

\subsection{Beach Bar}\label{app:beach_bar}

For implementation of the model described in the main text, we define a reward function $r \bc{X_{t}^i, U_{t}^i, \mathbb G^i_t} = \frac{2}{|\mathcal X|} \abs{B - X_{t}^i} + \frac{2}{|\mathcal X|} \abs{U_{t}^i} - 3 \mathbb G^i_t (x)$ and dynamics $X_{t+1}^i = X_{t}^i + U_{t}^i + \epsilon_t^i$, where $\epsilon_t^i$ is a random noise variable equal to $-1$ or $1$ with probabilities $0.05$ each, and $0$ otherwise.

\end{document}